\documentclass[12pt,a4paper]{article}

\usepackage[utf8]{inputenc}
\usepackage[T1]{fontenc}
\usepackage{lmodern}
\usepackage[english]{babel}
\usepackage{microtype}
\usepackage{setspace}
\usepackage{makecell}

\usepackage[a4paper,margin=1in]{geometry}
\usepackage[svgnames]{xcolor}

\definecolor{MainRed}{RGB}{139,0,0}
\definecolor{DarkRed}{RGB}{100,0,0}
\definecolor{SoftRed}{RGB}{178,72,72}
\definecolor{PaleRed}{RGB}{250,235,235}
\definecolor{TextGray}{RGB}{70,70,70}
\definecolor{LightGray}{RGB}{230,230,230}

\colorlet{square}{MainRed}
\colorlet{codered}{MainRed}

\definecolor{codegreen}{RGB}{46,111,64}
\definecolor{codegray}{RGB}{105,105,105}
\definecolor{codepurple}{RGB}{122,62,142}
\definecolor{backcolour}{RGB}{250,246,246}

\usepackage{amsmath,amssymb,amsthm,mathtools,bm}

\usepackage{graphicx}
\usepackage{booktabs}
\usepackage{float}
\usepackage{pdflscape}
\usepackage{caption}
\usepackage{subcaption}
\usepackage{placeins}

\usepackage{longtable}
\usepackage{tabularx}
\usepackage{array}
\usepackage{ragged2e}
\newcolumntype{Y}{>{\raggedright\arraybackslash}X}

\newcommand{\tablestretch}{\renewcommand{\arraystretch}{1.10}}

\newcommand{\compacttablecols}{\setlength{\tabcolsep}{3.5pt}}

\usepackage{tikz}
\usetikzlibrary{arrows.meta,fit,calc,positioning}

\usepackage[toc,page]{appendix}

\usepackage{natbib}
\setcitestyle{authoryear,open={(},close={)}}
\usepackage{bibunits}
\defaultbibliographystyle{apalike}

\theoremstyle{plain}
\newtheorem{thm}{Theorem}[section]
\newtheorem{prop}[thm]{Proposition}
\newtheorem{lem}[thm]{Lemma}
\newtheorem{cor}[thm]{Corollary}

\theoremstyle{definition}
\newtheorem{defn}[thm]{Definition}

\theoremstyle{remark}
\newtheorem{remark}[thm]{Remark}

\DeclareMathOperator{\Var}{Var}
\DeclareMathOperator{\Cov}{Cov}

\newcommand{\1}{\mathbf{1}}

\usepackage{xurl}
\usepackage{hyperref}
\hypersetup{
  colorlinks=true,
  linkcolor=MainRed,
  citecolor=MainRed,
  urlcolor=MainRed,
  anchorcolor=MainRed
}

\newcommand{\weqref}[1]{\begingroup\hypersetup{linkcolor=white}\eqref{#1}\endgroup}

\newif\ifshowonline
\showonlinefalse

\newcommand{\sidoi}{https://doi.org/10.5281/zenodo.21796549}
\newcommand{\siref}[1]{\ref{#1}}

\newcommand{\mt}{}

\newif\ifinsertfigmarkers
\insertfigmarkerstrue

\newlength{\figuregap}

\makeatletter
\newcommand{\definesinum}[2]{\expandafter\gdef\csname sinum@#1\endcsname{#2}}
\newcommand{\sinum}[1]{\@ifundefined{sinum@#1}{\textbf{??}}{\csname sinum@#1\endcsname}}
\makeatother

\definesinum{app:alternative_information_sets}{D.1}
\definesinum{app:calibration_invariance}{A.5}
\definesinum{app:orthogonality}{B}
\definesinum{app:proof_Moments}{A.1}
\definesinum{app:proof_Proposition_GIRF_PD}{A.2}
\definesinum{app:proof_pd_ar}{A.3}
\definesinum{app:proof_pd_es}{A.4}
\definesinum{app:satellite_bma}{D.3}
\definesinum{app:simulation_design}{C}
\definesinum{app:variables_detail}{D.2}
\definesinum{lem:closed}{A.1}
\definesinum{tab:direct_channel_pd}{6}
\definesinum{tab:direct_channel_satellite}{5}
\definesinum{tab:dralacbn_information_set_robustness}{7}
\definesinum{tab:var_stability_info_sets}{8}

\renewcommand{\siref}[1]{\href{\sidoi}{\sinum{#1}}}

\begin{document}

\begingroup
\singlespacing
\thispagestyle{plain}

\vspace*{1.7cm}

\begin{center}
{\Large Generalized Impulse Responses of Portfolio Default Probabilities: A Modular Framework with an Application to Geopolitical Risk
 \par}
\vspace{0.8cm}

Guillaume Flament$^{\S}$
\quad
Christophe Hurlin$^{\dagger,\ddagger,*}$
\quad
Quentin Lajaunie$^{\dagger,\S}$
\quad
Yoann Pull$^{\dagger,\S}$

\vspace{0.5cm}

{\normalsize June 2026}
\end{center}

\vspace{1.1cm}

\begin{center}
\textbf{Summary}
\end{center}

\begin{quote}
\small
Credit stress testing requires impulse responses of portfolio default probabilities, not only macro-financial drivers. We derive closed-form generalized impulse responses for the mean, quantiles (PD-at-Risk), and expected shortfall in a modular framework combining a Bayesian VAR, a Gaussian satellite, and the Merton–Vasicek model underlying Basel IRB regulation. Results extend to any probit-Gaussian mapping of a latent factor. Nonlinearity makes responses depend on conditional means and variances; plug-in evaluations understate projected default probability levels by 6-8\% and miss tail quantiles. For U.S. geopolitical risk shocks, 99\%-quantile responses exceed mean responses by 50\%, and peak responses vary 4.6-fold across the credit cycle.

\end{quote}

\vspace{0.2cm}

\noindent\textbf{Keywords:} generalized impulse responses, probit models, portfolio default probabilities, Bayesian VAR, credit stress testing, geopolitical risk.

\vspace{0.2cm}

\noindent\textbf{JEL codes:} C11, C32, E44, G01, G21, G32.

\vfill

\noindent\rule{\textwidth}{0.4pt}

\vspace{0.1cm}

\scriptsize
\noindent$^{\dagger}$University of Orl\'eans, Rue de Blois, 45067 Orl\'eans, France.\\
$^{\ddagger}$Institut Universitaire de France (IUF), 75231 Paris, France.\\
$^{\S}$Square Research Center, 173 Av.\ Achille Peretti, 92200 Neuilly-sur-Seine, France.
\smallskip

\noindent$^{*}$Corresponding author: Christophe Hurlin, University of Orl\'eans, Rue de Blois, 45067 Orl\'eans, France.\\
Email: \href{mailto:christophe.hurlin@univ-orleans.fr}{christophe.hurlin@univ-orleans.fr}.

\smallskip

\noindent\textbf{Author emails:}\\
Guillaume Flament, \href{mailto:guillaume.f.flament@hotmail.fr}{guillaume.f.flament@hotmail.fr};\\
Christophe Hurlin, \href{mailto:christophe.hurlin@univ-orleans.fr}{christophe.hurlin@univ-orleans.fr};\\
Quentin Lajaunie, \href{mailto:quentin_lajaunie@hotmail.fr}{quentin\_lajaunie@hotmail.fr};\\
Yoann Pull, \href{mailto:yoann.pull.pro@gmail.com}{yoann.pull.pro@gmail.com}.

\smallskip

\noindent\textbf{Acknowledgments.}
We are grateful for the valuable insights, feedback, and discussions received throughout the development of this paper. We also thank participants at the 19th Financial Risks International Forum, held in Paris in March 2026, and at the 18th Annual SoFiE Conference, held in Macau in June 2026, for their helpful comments.

\smallskip

\noindent\textbf{Conflict of interest.}
The authors declare no conflict of interest.

\smallskip

\noindent\textbf{Supporting information.}
Proofs, additional derivations, robustness checks, and implementation details are
provided in the Supporting Information, archived at \url{\sidoi}.

\smallskip

\noindent\textbf{Data availability statement.}
The data and code required to reproduce the results of this study are archived at
\url{https://doi.org/10.5281/zenodo.21797928}. All input series are publicly
available from the sources documented in the repository.
\endgroup

\doublespacing
\setlength{\parskip}{\medskipamount}

\section{Introduction}

In many macro-financial applications the object of interest is not a variable of the econometric model itself, but a nonlinear transformation of a latent Gaussian state produced by a separate modeling block. Credit stress testing is a leading example, in which a macro-financial model generates scenarios and a separate satellite equation maps them into a latent systematic credit factor and, through it, into portfolio default probabilities. The generalized impulse response of this default-probability measure need not coincide with that of the macro-financial variables, because a nonlinear map depends on the entire conditional distribution of the factor, not only on its mean path. The ``plug-in'' convention used in practice, which evaluates the map along the expected scenario path, reproduces average responses reasonably well but understates default-probability \emph{levels} and cannot recover their upper-tail \emph{quantiles}.

This paper develops closed-form mean and upper-tail generalized impulse responses of portfolio default probabilities in such modular systems combining a Bayesian vector autoregression, a Gaussian satellite equation for the latent credit factor, and a credit-risk map. The leading credit-risk map is the Merton--Vasicek asymptotic single-risk-factor (ASRF) model \citep{merton1974pricing,vasicek2002distribution,Gordy2003}, the structural portfolio model underlying the Basel internal-ratings-based (IRB) capital function. Since this map is a Gaussian probit transformation of the latent systematic factor $Z$, the analytical results can be stated for the broader class of probit-Gaussian maps $f(Z)=\Phi(a+bZ)$. It makes two contributions. First, we derive the response of the portfolio's \emph{mean} default probability in closed form, integrating the map over the conditional distribution of the latent factor; the expression separates a location channel from a variance channel and quantifies how far the plug-in convention of supervisory practice understates default-probability \emph{levels}. Second, we show that the same conditional moments characterize the upper tail of the portfolio default-probability distribution. They yield closed-form responses of its \emph{quantiles}, which we term \emph{PD-at-Risk} by analogy with the Growth-at-Risk of \citet{adrian2019vulnerable}, and of its expected shortfall, a coherent tail-risk measure \citep{acerbi2002coherence,wang2021axiomatic} adopted in the Basel framework to better capture tail risk under stress \citep{BCBS2019}. The resulting mean, quantile, and expected-shortfall responses are exact, require no simulation of default events, and accommodate macro-financial and credit blocks estimated on samples of different lengths. Section~\ref{sec:simulation_benchmark} benchmarks the closed-form expressions against forward simulation. Throughout, U.S. geopolitical risk serves as the leading illustration, but the framework applies to any macro-financial innovation.

This modular structure is a constraint imposed by practice, not a modeling convenience, and its two outputs map to the two pillars of credit-risk regulation. Expected-credit-loss accounting under IFRS~9 and CECL ties provisions to expected loss, and hence to default-probability \emph{levels}. Regulatory capital instead targets a high \emph{quantile} of the conditional default probability, the PD-at-Risk object above, of which the Basel IRB charge is the $99.9\%$ case \citep{Gordy2003}. Its adequacy is assessed under severe but plausible scenarios, as in the ICAAP in the European Union and the CCAR and DFAST exercises in the United States. In both regimes these mappings are carried by satellite models that link the credit-risk factor to the macro-financial variables and are embedded in banks' provisioning and capital infrastructures \citep{HenryKok2013}. In supervisory exercises these satellites are pre-existing and governed, and cannot generally be changed mid-exercise.\footnote{See \citet[paras.~125--129]{EBA2025method} on internal satellite models, documentation requirements, and restrictions on changing the initial modelling approach during the EU-wide stress test.}

The econometric task is therefore not to replace the satellite with a fully joint macro-credit system re-estimated for each shock, but to attach a dynamic macro-financial model to this existing mapping. This modular separation is also empirically convenient. Macro-financial series span decades, whereas default histories are shorter and shift with regulation, accounting, and portfolio composition. The modular design then propagates the shock on the full macro-financial history while re-estimating only the credit block on the available default window.

The closest literature is macro-to-credit stress testing, where satellite equations map macroeconomic scenarios into default rates, loss rates, or portfolio risk measures \citep{Wilson1997,Virolainen2004,pesaran2006macroeconomic,camara2015mercure}. These satellite architectures underpin both supervisory and internal stress-testing systems \citep{Quagliariello2009,HenryKok2013,borio2014stress}. Our contribution is to keep this modular architecture but to replace plug-in propagation and default-event simulation with an exact analytical response of the credit-risk measure itself, a generalized impulse response of portfolio default probabilities, not only of the macro-financial variables that enter the satellite.

The paper also relates to dynamic models of event probabilities and nonlinear impulse-response analysis \citep{koop1996impulse,pesaranShin1998}, in which latent Gaussian representations are a standard device for threshold probabilities \citep{albert1993bayesian,dueker2005dynamic,ChanPfarrhofer2025}. Relative to this literature, our object is not a probability \emph{forecast} but a generalized impulse \emph{response}: the difference between a shocked and a baseline conditional probability, derived for both the \emph{mean} and the \emph{quantiles} of portfolio default probabilities, with an explicit variance channel and across blocks of different sample lengths. The closest antecedent, \citet{FornariLemke2010}, integrates a probit over the predictive distribution of a latent Gaussian state to obtain conditional recession probabilities; we share that integration logic but deliver an impulse response, in the mean and in the tail, rather than a level forecast. The question thus shifts from ``how likely is the event'' to ``how does a macro-financial innovation move the entire distribution of portfolio default risk''.

Finally, the paper connects to recent work on stressed scenarios and distributional risk measures \citep{gonzalezrivera2024expecting,ChavleishviliManganelli2024}. In quantile vector autoregressions, the propagation mechanism itself may vary across conditional quantiles. In our framework, macro-financial propagation is governed by the Bayesian VAR, while the distributional responses arise from the conditional distribution of the latent systematic credit factor and its monotone transformation into default probabilities.

We apply the framework to U.S. credit risk under geopolitical-risk shocks, measured by the Geopolitical Risk Index (GPR) of \citet{caldara2022geopolitical}, a newspaper-based index that has proved informative for identifying the economic effects of geopolitical events \citep{laudati2023identifying}. Geopolitical risk is a natural illustration for a tail-oriented framework: it has become a supervisory stress-testing priority \citep{ECB_SSM_PR_2025}, it materializes through discrete, potentially large events, and it is difficult to represent through standard macroeconomic scenarios. We study the forward response to a given innovation, complementing the reverse scenario-design problem studied in \citet{hurlin2026reverse}. The macro-financial block is estimated on quarterly U.S. data from 1986:Q1 to 2024:Q4, and the credit-risk block uses the delinquency rate on all loans and leases at U.S. commercial banks as a long aggregate proxy for portfolio default risk.

The empirical results concern the entire conditional distribution of portfolio default probabilities, not only its mean. A one-standard-deviation geopolitical-risk innovation raises the mean portfolio default probability by $0.033$ percentage points at its three-quarter peak, about $1\%$ of the through-the-cycle level. The same innovation moves the upper tail substantially more: the $99\%$ default-probability quantile rises about $1.5$ times the mean response. The response is also strongly state-dependent: across the historically observed range of credit conditions, the peak response varies by a factor of $4.6$, and is largest when the shock strikes the most stressed credit states. A comparison with the perfect-foresight convention of supervisory stress tests shows that treating the scenario path as known leaves impulse responses almost unchanged but understates projected default-probability \emph{levels} by about seven to eight percent at the three-year horizon, through the convexity of the credit-risk map in the empirically relevant low-PD region.

The rest of the paper is organized as follows. Section~\ref{sec:framework} presents the econometric framework, derives the closed-form mean, quantile, and expected-shortfall responses of default probabilities, and benchmarks them against forward simulation. Section~\ref{sec:empirical_analysis} applies it to U.S. geopolitical risk through four exercises: a standard impulse-response analysis, a historical-episode analysis of state dependence, a short-default-sample exercise, and a comparison with the perfect-foresight convention. Section~\ref{sec:conclusion} concludes.

\section{General framework}
\label{sec:framework}

We consider a macro-financial system containing a variable of interest and a
set of variables that describe aggregate economic and financial conditions.
Let $R_t$ denote the variable of interest and let
$\mathbf X_t=(X_{1,t},\ldots,X_{n-1,t})^\top$ collect the remaining
macro-financial variables. We stack the variables as
\begin{equation}
    Y_t=(R_t,\mathbf X_t^\top)^\top \in \mathbb R^n.
\end{equation}
The first element of $Y_t$ is the variable whose innovation defines the shock
of interest.\footnote{Placing the variable of interest first is a notational
convention; the generalized impulse responses derived below are invariant to the
ordering of the remaining variables. The identifying content of treating its
reduced-form innovation as the structural shock is discussed in
Section~\ref{sec:irf_transmission}.}

The object of interest is the dynamic effect of a reduced-form innovation in
$R_t$ on the default probability of a credit portfolio. Formally, writing
$f(Z_{t+h})$ for the conditional event probability implied by a
probit observation block with conditionally Gaussian systematic factor $Z_t$
(both defined below), $u_{gt}$ for the innovation to the equation of the variable of
interest, and $\Omega_{t-1}$ for the VAR information set, we study the
generalized impulse response (GIRF)
\begin{equation}
    \psi_{f(Z)}^g(h,\delta_g,\omega_{t-1})
    =
    \mathbb E\!\left[
        f(Z_{t+h})\mid u_{gt}=\delta_g,\Omega_{t-1}=\omega_{t-1}
    \right]
    -
    \mathbb E\!\left[
        f(Z_{t+h})\mid \Omega_{t-1}=\omega_{t-1}
    \right].
    \label{eq:GIRF_PD_Def}
\end{equation}
The leading credit-risk specification for $f$ is the
Merton–Vasicek ASRF map used in the Basel IRB framework
(Section~\ref{sec:mv_instance}). We nevertheless derive the results for the
broader class of probit transformations of a Gaussian latent factor, of which
the Merton–Vasicek model is a particular case.
Because the probit map is nonlinear, this response is not obtained by
evaluating the map at the mean response of the systematic factor $Z_{t+h}$.

The framework is summarized in Figure~\ref{fig:twostep_approach}. It combines two separately estimated blocks: a macro-financial VAR and a credit-risk satellite linked through a subset of current and lagged macro-financial variables, denoted $Y_t^{(s)}$. The architecture is deliberately recursive: macro-financial conditions drive credit risk, while credit risk does not feed back into the macro-financial block. This mirrors internal and supervisory stress-testing systems and keeps the credit-risk bridge separate from the macro-financial model.

The two blocks may be estimated over different sample lengths. Macro-financial variables are observed over $t=1,\ldots,T_Y$, whereas portfolio default rates are available only over $t=T_0,\ldots,T_Y$, with $T_0>1$ and $T_d=T_Y-T_0+1$. The VAR is estimated on the full macro-financial sample, the latent factor is reconstructed from the shorter default sample, and the satellite is estimated on their overlap.

This sample separation concerns estimation only. The shock is a VAR innovation, and the GIRFs of $Y_{t+h}$, $Z_{t+h}$, and $f(Z_{t+h})$ are all defined conditional on the same VAR information set. As shown in Figure~\ref{fig:twostep_approach}, the innovation propagates through the VAR, the satellite, and the Merton–Vasicek map to the mean and upper tail of the portfolio default-probability distribution.

\begin{figure}[!htbp]
\centering
\resizebox{\textwidth}{!}{%
\begin{tikzpicture}[
  font=\small,
  >=Latex,
  tag/.style={rounded corners=6pt, draw=none, text=white, align=center,
              minimum width=3.9cm, minimum height=1.15cm, inner sep=5pt,
              font=\small},
  mainbox/.style={rounded corners=10pt, draw=none, text=white, align=center,
               minimum width=3.9cm, minimum height=1.75cm, inner sep=7pt},
  obox/.style={rounded corners=8pt, draw=none, text=white, align=center,
              minimum width=3.6cm, minimum height=0.95cm, inner sep=5pt},
  titlebadge/.style={draw=MainRed!45, fill=PaleRed, text=DarkRed,
                     rounded corners=5pt, inner xsep=10pt, inner ysep=5pt,
                     font=\bfseries},
  frame/.style={draw=MainRed!25, rounded corners=18pt, line width=0.9pt,
                inner sep=14pt},
  feed/.style={-Latex, line width=1.0pt, draw=MainRed!55},
  flow/.style={-Latex, line width=1.7pt, draw=DarkRed}
]
\def\dx{4.9cm}

\node (tY) at (0,2.6)     [tag, fill=SoftRed] {Macro-financial\\ sample $Y_t$~\weqref{eq:var}};
\node (tD) at (\dx,2.6)   [tag, fill=SoftRed] {Default sample $\{d_t\}$\\ factor $Z_t$~\weqref{eq:reconstruct_general}};

\node (bVAR) at (0,0)     [mainbox, fill=MainRed]
  {\textbf{Macro-financial VAR}\\[2mm] response $\psi_Y^{g}(h)$~\weqref{eq:GIRF-Y-closed}};
\node (bSat) at (\dx,0)   [mainbox, fill=MainRed]
  {\textbf{Satellite equation}\\[2mm] factor response\\ $\psi_Z^{g}(h)$~\weqref{eq:GIRF-Z}};
\node (bMap) at (2*\dx,0) [mainbox, fill=DarkRed]
  {\textbf{Merton--Vasicek ASRF}\\[2mm] $\boldsymbol{\pi(Z)}$~\weqref{eq:pd_conditional_z}};
\node[text=DarkRed, font=\footnotesize\itshape, align=center] (mvlab)
  at ($(bMap.south)+(0,-0.55cm)$) {extends to any probit map\\ $f(Z)=\Phi(a+bZ)$};

\node (oMean) at (3*\dx,1.25) [obox, fill=DarkRed] {Mean~\weqref{eq:GIRF_PD}};
\node (oAR)   at (3*\dx,0)     [obox, fill=DarkRed] {PD-at-Risk~\weqref{eq:pd_ar_girf}};
\node (oES)   at (3*\dx,-1.25) [obox, fill=DarkRed] {Expected shortfall~\weqref{eq:pd_es_girf}};

\draw[feed] (tY.south) -- (bVAR.north);
\draw[feed] (tD.south) -- (bSat.north);

\draw[-{Latex[length=4.2mm,width=3.4mm]}, line width=2.8pt, draw=MainRed]
   ($(bVAR.west)+(-2.0cm,0)$) -- (bVAR.west)
   node[midway,above=1pt,text=MainRed,font=\bfseries] {Shock $\delta_g$};
\draw[flow] (bVAR.east) -- (bSat.west);
\draw[flow] (bSat.east) -- (bMap.west);
\draw[flow] (bMap.east) -- (oMean.west);
\draw[flow] (bMap.east) -- (oAR.west);
\draw[flow] (bMap.east) -- (oES.west);

\end{tikzpicture}%
}
\caption{\textbf{A modular VAR--Merton framework.}
{\footnotesize\textit{Notes: A macro-financial VAR is estimated from the
macro-financial sample and a latent Gaussian factor $Z_t$ is reconstructed from
the default sample and linked to the VAR through a satellite equation. A
reduced-form innovation $\delta_g$ propagates from left to right: its impulse
response on $Y_t$ feeds the satellite to give the factor response $\psi_Z^{g}$,
which the Merton--Vasicek ASRF map turns into closed-form responses of the
default-probability distribution: its mean, its quantiles (PD-at-Risk), and
its expected shortfall. The closed forms hold for any probit map
$f(Z)=\Phi(a+bZ)$ of a Gaussian latent factor.}}}
\label{fig:twostep_approach}
\end{figure}
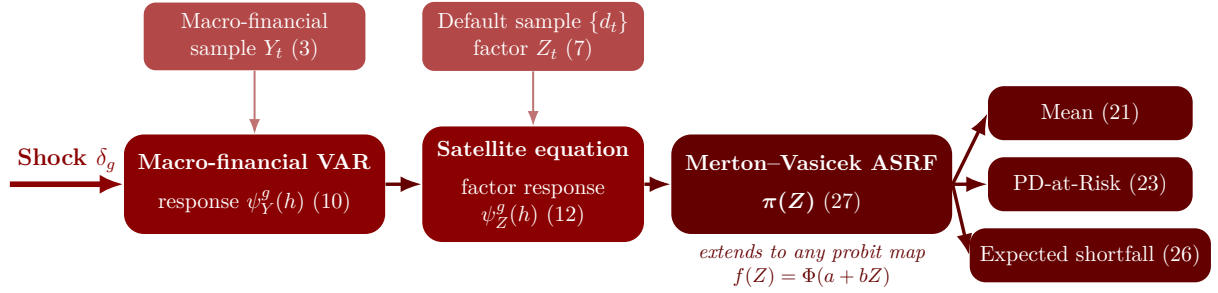

\subsection{Macro-credit architecture}
\label{sec:macro_credit_architecture}

\subsubsection{Macro-financial block}

The macro-financial block is a Gaussian VAR($P$) for $Y_t$:
\begin{equation}
    Y_t = c + \sum_{i=1}^{P} A_i Y_{t-i} + u_t,\qquad
    u_t \stackrel{\text{i.i.d.}}{\sim} \mathcal N(0,\Sigma_u),
    \label{eq:var}
\end{equation}
where $c\in\mathbb R^n$, $A_i\in\mathbb R^{n\times n}$, and
$\Sigma_u\succ 0$. The process is assumed covariance-stationary, with
moving-average representation
\begin{equation}
    Y_t=\mu_Y+\sum_{h=0}^{\infty}\Psi_h u_{t-h},
    \qquad \Psi_0=I_n.
\end{equation}
All derivations in this section are stated conditional on the parameters of
the two blocks, $(c,A_1,\ldots,A_P,\Sigma_u)$ and
$(\beta_0,\beta,\sigma_\eta^2)$ below; estimation uncertainty is introduced in
Section~\ref{sec:posterior_implementation}.

\subsubsection{Credit-risk block}

The second block maps the macro-financial state into the probability of a
threshold-type event. Let $Z_t$ be a scalar latent index that is conditionally
Gaussian given the macro-financial information set, oriented so that higher
values correspond to a lower event probability. The object of interest is a
\emph{probit} transformation of this index,
\begin{equation}
    f(Z_t)=\Phi(a+bZ_t),
    \qquad a\in\mathbb R,\quad b<0,
    \label{eq:probit_map_general}
\end{equation}
where $\Phi(\cdot)$ is the standard normal cumulative distribution function and
the pair $(a,b)$ is fixed by the application.\footnote{The probit-Gaussian pairing is what keeps the responses in closed form: for Gaussian $Z$, $\mathbb E[\Phi(a+bZ)]$ is again a probit (Lemma~\siref{lem:closed}), whereas a logistic link admits no such elementary expression.}

The index is linked to the macro-financial variables through a Gaussian
satellite equation. Let $Y_t^{(s)}$ collect the subset of current and lagged VAR
variables that enter the satellite. We specify
\begin{equation}
    Z_t = \beta_0+\beta^\top Y_t^{(s)}+\eta_t,
    \qquad
    \eta_t\sim\mathcal N(0,\sigma_\eta^2),
    \qquad
    \mathbb E[\eta_t\mid Y_t^{(s)}]=0,
    \label{eq:Z}
\end{equation}
and assume that $\{\eta_t\}$ is serially independent and independent of the
VAR innovation sequence $\{u_s\}_{s\in\mathbb Z}$.\footnote{We discuss this
assumption in Appendix~\siref{app:orthogonality}.} The independence of $\eta_t$ and
$u_t$ is an exclusion restriction: it rules out any effect of the innovation of
interest on the latent index beyond the macro-financial channel $Y_t^{(s)}$.

How the index is measured determines how $(\beta,a,b)$ are estimated. When the
event is observed only as a binary outcome, the parameters are estimated jointly
by probit maximum likelihood. When instead a continuous outcome $d_t$ with
$\mathbb E[d_t\mid Z_t]=f(Z_t)$ is observed (as with aggregate portfolio
default rates in the credit application), the index can be reconstructed by
inverting the map,
\begin{equation}
    Z_t=\frac{\Phi^{-1}(d_t)-a}{b},
    \label{eq:reconstruct_general}
\end{equation}
and the satellite \eqref{eq:Z} estimated by least squares. We follow this second
route, with $(a,b)$ given by the Merton--Vasicek calibration of
Section~\ref{sec:mv_instance}.

\subsection{Transmission to default risk}
\label{sec:irf_transmission}

We now derive the three steps of the lower panel of
Figure~\ref{fig:twostep_approach}: the macro-financial GIRF, its transmission
to the systematic factor, and the closed-form default-probability response
\eqref{eq:GIRF_PD_Def}. Let $g\in\{1,\ldots,n\}$ denote the index of the
variable of interest, $e_g$ the $g$-th canonical unit vector, and
$u_{gt}=e_g^\top u_t$. The shock is a reduced-form innovation to equation $g$,
conditional on the information set $\Omega_{t-1}$ generated by the VAR history
up to $t-1$, with realization $\omega_{t-1}$. Throughout, we use the
convention $\psi_Y^g(h',\cdot,\cdot)=0$ for $h'<0$.

\subsubsection{Macro-financial generalized impulse responses}

\begin{defn}[\textbf{Generalized impulse response}]
\label{def:GIRF}
For a horizon $h\geq0$, a history $\omega_{t-1}$, and a scalar innovation
$\delta_g$ to equation $g$, the generalized impulse response of $Y_t$ is
\begin{equation}
    \psi_Y^g(h,\delta_g,\omega_{t-1})
    =
    \mathbb E\!\left[
        Y_{t+h}\mid u_{gt}=\delta_g,\Omega_{t-1}=\omega_{t-1}
    \right]
    -
    \mathbb E\!\left[
        Y_{t+h}\mid \Omega_{t-1}=\omega_{t-1}
    \right].
    \label{eq:GIRF-def-interest}
\end{equation}
\end{defn}

In the linear Gaussian VAR, this response does not depend on the history
$\omega_{t-1}$: Gaussian conditioning gives
$\mathbb E[u_t\mid u_{gt}=\delta_g]=\Sigma_u e_g\,\delta_g/\sigma_{gg}$ with
$\sigma_{gg}=e_g^\top\Sigma_u e_g$, so that
\begin{equation}
    \psi_Y^g(h,\delta_g)
    =
    \Psi_h
    \frac{\Sigma_u e_g}{\sigma_{gg}}\delta_g,
    \label{eq:GIRF-Y-general}
\end{equation}
and, for a one-standard-deviation innovation
$\delta_g=\sqrt{\sigma_{gg}}$,
\begin{equation}
    \psi_Y^g(h)
    =
    \sigma_{gg}^{-1/2}\Psi_h\Sigma_u e_g.
    \label{eq:GIRF-Y-closed}
\end{equation}
This is the generalized impulse response of \citet{koop1996impulse}. In the
linear Gaussian VAR, it coincides numerically with the orthogonalized impulse
response obtained from a Cholesky decomposition in which variable $g$ is ordered
first \citep[Proposition~3.1]{pesaranShin1998}. The generalized impulse response
is invariant to the ordering of the remaining variables. Interpreting the
reduced-form innovation $u_{gt}$ as the structural shock of interest does,
however, maintain that variable $g$ is contemporaneously exogenous, in the sense
that it does not respond within the period to innovations in the other variables.
We take this contemporaneous exogeneity as a maintained identifying restriction,
plausible for a slow-moving, event-based index\footnote{Such as the Geopolitical Risk (GPR) index of \citet{caldara2022geopolitical} used in the empirical application (Section~\ref{sec:empirical_analysis}).} at the quarterly frequency,
and read the responses accordingly.

\subsubsection{Transmission to the systematic factor}

The satellite equation may include selected variables and lags from the VAR.
We represent this selection by a linear lag operator. Let $L$ denote the lag
operator, $L^\ell x_t=x_{t-\ell}$, and for a maximum lag
$L_{\max}\in\mathbb N_0$ define
\begin{equation}
    S^{(s)}(L)
    =
    \sum_{\ell=0}^{L_{\max}} S_\ell^{(s)}L^\ell,
    \qquad
    S_\ell^{(s)}\in\{0,1\}^{m\times n},
    \qquad
    Y_t^{(s)}=S^{(s)}(L)Y_t\in\mathbb R^m,
    \label{eq:selection_operator}
\end{equation}
where each row of $\{S_\ell^{(s)}\}_{\ell=0}^{L_{\max}}$ selects exactly one
variable at one lag. In the empirical implementation, the variable of interest
does not enter the satellite equation: the $g$-th column of each
$S_\ell^{(s)}$ is zero.

By linearity, generalized impulse responses propagate through the selection
operator and the satellite equation. The intercept is common to the shocked and
baseline forecasts, and $\eta_t$ is independent of the VAR innovations
$\{u_s\}$, so its generalized response to $u_{gt}$ is zero; both therefore drop
out of the difference between shocked and baseline forecasts, and the
systematic-factor response is
\begin{equation}
    \psi_Z^g(h,\delta_g,\omega_{t-1})
    =
    \beta^\top S^{(s)}(L)
    \psi_Y^g(h,\delta_g,\omega_{t-1})
    =
    \sum_{\ell=0}^{L_{\max}}
        \beta^\top S_\ell^{(s)}
        \psi_Y^g(h-\ell,\delta_g,\omega_{t-1}).
    \label{eq:GIRF-Z}
\end{equation}

\subsubsection{Conditional moments of the systematic factor}

The closed-form response \eqref{eq:GIRF_PD_Def} requires the conditional
distribution of $Z_{t+h}$ under the baseline forecast and under the scalar
innovation. Both are Gaussian, hence characterized by
\begin{equation}
    \mu_{t+h}
    =
    \mathbb E[Z_{t+h}\mid\Omega_{t-1}=\omega_{t-1}],
    \qquad
    \mu_{t+h}^{(\delta_g)}
    =
    \mathbb E[Z_{t+h}\mid u_{gt}=\delta_g,\Omega_{t-1}=\omega_{t-1}],
    \label{eq:mu-definitions}
\end{equation}
\begin{equation}
    s_{t+h}^2
    =
    \Var[Z_{t+h}\mid\Omega_{t-1}=\omega_{t-1}],
    \qquad
    \big(s_{t+h}^{(\delta_g)}\big)^2
    =
    \Var[Z_{t+h}\mid u_{gt}=\delta_g,\Omega_{t-1}=\omega_{t-1}].
    \label{eq:var-definitions}
\end{equation}
Conditioning on $u_{gt}=\delta_g$ fixes one component of the contemporaneous
innovation vector; the remaining components stay random with conditional
covariance matrix
\begin{equation}
    \Sigma_{u\mid g}
    =
    \Var(u_t\mid u_{gt})
    =
    \Sigma_u
    -
    \frac{\Sigma_u e_g e_g^\top \Sigma_u}{\sigma_{gg}}.
    \label{eq:sigma-u-cond-g}
\end{equation}
For $h\geq0$ and $0\leq q\leq h$, define
\begin{equation}
    G_{h,q}
    =
    \sum_{\ell=0}^{L_{\max}}
        S_\ell^{(s)}\Psi_{h-\ell-q}
    \in\mathbb R^{m\times n},
    \qquad
    B(h,q)
    =
    \beta^\top G_{h,q}
    \in\mathbb R^{1\times n},
    \label{eq:G-hq}
\end{equation}
with the convention $\Psi_r=0$ for $r<0$.

\begin{prop}[\textbf{Conditional moments of the systematic factor}]
\label{prop:Moments-Z}
Under \eqref{eq:var}--\eqref{eq:Z}, the baseline conditional moments of
$Z_{t+h}$ are
\begin{equation}
    \mu_{t+h}
    =
    \beta_0
    +
    \sum_{\ell=0}^{L_{\max}}
        \beta^\top S_\ell^{(s)}
        \mu_{t+h-\ell}^{Y},
    \qquad
    \mu_{t+h-\ell}^{Y}
    =
    \mathbb E[Y_{t+h-\ell}\mid\Omega_{t-1}=\omega_{t-1}],
    \label{eq:mu-baseline}
\end{equation}
\begin{equation}
    s_{t+h}^2
    =
    \sum_{q=0}^{h}
        B(h,q)\Sigma_u B(h,q)^\top
    +
    \sigma_\eta^2.
    \label{eq:s2-baseline}
\end{equation}
Under the scalar innovation $u_{gt}=\delta_g$, the shocked conditional moments
are
\begin{equation}
    \mu_{t+h}^{(\delta_g)}
    =
    \mu_{t+h}
    +
    \psi_Z^g(h,\delta_g,\omega_{t-1}),
    \label{eq:mu-shocked}
\end{equation}
\begin{equation}
    \big(s_{t+h}^{(\delta_g)}\big)^2
    =
    B(h,0)\Sigma_{u\mid g}B(h,0)^\top
    +
    \sum_{q=1}^{h}
        B(h,q)\Sigma_u B(h,q)^\top
    +
    \sigma_\eta^2,
    \label{eq:s2-shocked}
\end{equation}
with $\psi_Z^g$ given in \eqref{eq:GIRF-Z}.
\end{prop}

The derivation is in Appendix~\siref{app:proof_Moments}. The scalar innovation
shifts the conditional mean through the systematic-factor GIRF and changes the
contemporaneous ($q=0$) contribution to the conditional variance through
$\Sigma_{u\mid g}$, while future VAR innovations remain governed by
$\Sigma_u$.

\subsubsection{Closed-form default-probability response}

We can now state the main analytical result: the default-probability response
\eqref{eq:GIRF_PD_Def} is an explicit function of the four conditional moments
of Proposition~\ref{prop:Moments-Z}.

\begin{prop}[\textbf{Closed-form GIRF of portfolio default probabilities}]
\label{prop:GIRF-PD}
Under \eqref{eq:var}--\eqref{eq:Z}, the generalized impulse response of the
event probability $f(Z_{t+h})=\Phi(a+bZ_{t+h})$ to a scalar innovation
$\delta_g$ is
\begin{equation}
    \psi_{f(Z)}^g(h,\delta_g,\omega_{t-1})
    =
    \Phi\!\left(
        \frac{a+b\,\mu_{t+h}^{(\delta_g)}}
        {\sqrt{1+b^2\big(s_{t+h}^{(\delta_g)}\big)^2}}
    \right)
    -
    \Phi\!\left(
        \frac{a+b\,\mu_{t+h}}
        {\sqrt{1+b^2 s_{t+h}^2}}
    \right).
    \label{eq:GIRF_PD}
\end{equation}
\end{prop}

Appendix~\siref{app:proof_Proposition_GIRF_PD} establishes
Proposition~\ref{prop:GIRF-PD}.

\begin{remark}
The closed form in \eqref{eq:GIRF_PD} is related to the
recession-probability forecasts of \citet{FornariLemke2010}, which also
integrate over the predictive distribution of a latent Gaussian state. The
object considered here is a generalized impulse response, it compares the
default probability conditional on a realized innovation with its baseline
conditional value. This comparison may change both the conditional mean and the
conditional variance of the systematic factor, yielding a location channel and,
when the variances differ, a variance channel.
\end{remark}

Substituting \eqref{eq:mu-shocked} into \eqref{eq:GIRF_PD} expresses the response directly
in terms of the systematic-factor GIRF, making explicit that the innovation
operates through two channels: it shifts the location of the conditional
distribution by $\psi_Z^g(h)$, and it tightens its dispersion at impact
through $\Sigma_{u\mid g}$. Evaluating the probit map $\pi$ only along the
mean response, the \emph{plug-in} convention of supervisory practice, amounts to
setting the conditional variances in \eqref{eq:GIRF_PD} to zero, which leaves the
response $f(\mu_{t+h}^{(\delta_g)})-f(\mu_{t+h})$. This approximation ignores
both the dispersion of the systematic factor and the way this dispersion changes
after the shock.\footnote{We quantify this bias in
Section~\ref{sec:perfect_foresight} of the empirical application.}

The same conditional moments deliver closed-form quantiles of the portfolio
default-probability distribution. We refer to these quantiles as
PD-at-Risk (PD-aR), in the spirit of the at-risk terminology of
\citet{adrian2019vulnerable}. Unlike quantile impulse responses in QVAR
models, where macro-financial propagation may vary across quantiles
\citep{ChavleishviliManganelli2024}, quantiles here arise from the conditional
distribution of the latent systematic factor and its monotone transformation
into default probabilities.

\begin{cor}[\textbf{Quantile response of portfolio default probabilities}]
\label{cor:pd_ar}
Under \eqref{eq:var}--\eqref{eq:Z}, for a confidence level $\alpha\in(0,1)$,
the $\alpha$-level PD-at-Risk (PD-aR) at horizon $h$ is
\begin{equation}
\mathrm{PD}^{\mathrm{aR}}_\alpha(h)
=
\Phi\!\left(
a+b\big(\mu_{t+h}-\Phi^{-1}(\alpha)\,s_{t+h}\big)
\right),
\label{eq:pd_ar_level}
\end{equation}
and its generalized impulse response to a scalar innovation $\delta_g$ is
\begin{equation}
\psi_{\mathrm{aR},\alpha}^g(h,\delta_g,\omega_{t-1})
=
\Phi\!\left(
a+b\big(\mu_{t+h}^{(\delta_g)}-\Phi^{-1}(\alpha)\,s_{t+h}^{(\delta_g)}\big)
\right)
-
\mathrm{PD}^{\mathrm{aR}}_\alpha(h).
\label{eq:pd_ar_girf}
\end{equation}
\end{cor}

Appendix~\siref{app:proof_pd_ar} gives the derivation. The mean response
\eqref{eq:GIRF_PD} integrates the conditional variance through the denominator,
whereas the quantile response shifts the conditional factor quantile in the
numerator; it is an exact monotone transform of the corresponding quantile of
$Z_{t+h}$.\footnote{Because $\pi$ is monotone, the PD-aR is a quantile of the
\emph{model-implied} conditional distribution of default probabilities. It is
obtained by applying the probit map $\pi$ to a Gaussian quantile of
$Z_{t+h}$, not to an empirical default-rate quantile. Its tail behaviour
therefore reflects the curvature of the map and the
Gaussianity of the latent factor.}

The PD-aR is a value-at-risk and is not subadditive. The same conditional
moments deliver in closed form the corresponding expected shortfall (ES), a coherent
risk measure \citep{acerbi2002coherence,wang2021axiomatic} and the tail statistic targeted by the
regulatory shift to expected shortfall \citep{BCBS2019}.

\begin{cor}[\textbf{Expected-shortfall response of portfolio default
probabilities}]
\label{cor:pd_es}
Under \eqref{eq:var}--\eqref{eq:Z}, for a confidence level $\alpha\in(0,1)$, the
$\alpha$-level expected shortfall of the portfolio default probability,
$\mathrm{ES}_\alpha(h)=\mathbb E[f(Z_{t+h})\mid
f(Z_{t+h})\ge\mathrm{PD}^{\mathrm{aR}}_\alpha(h)]$, is
\begin{equation}
\mathrm{ES}_\alpha(h)
=
\frac{1}{1-\alpha}\,
\Phi_2\!\left(m_{t+h},\,\Phi^{-1}(1-\alpha)\,;\,\rho^\star_{t+h}\right),
\label{eq:pd_es_level}
\end{equation}
where $\Phi_2(\cdot,\cdot;\rho^\star)$ is the standard bivariate normal CDF with
correlation $\rho^\star$ and
\begin{equation}
m_{t+h}
=
\frac{a+b\,\mu_{t+h}}
{\sqrt{1+b^2\,s_{t+h}^2}},
\qquad
\rho^\star_{t+h}
=
\frac{-b\,s_{t+h}}{\sqrt{1+b^2\,s_{t+h}^2}}.
\label{eq:pd_es_coeffs}
\end{equation}
Its generalized impulse response to a scalar innovation $\delta_g$ is
\begin{equation}
\psi^g_{\mathrm{ES},\alpha}(h,\delta_g,\omega_{t-1})
=
\frac{1}{1-\alpha}\left[
\Phi_2\!\left(m^{(\delta_g)}_{t+h},\Phi^{-1}(1-\alpha);\rho^{\star(\delta_g)}_{t+h}\right)
-
\Phi_2\!\left(m_{t+h},\Phi^{-1}(1-\alpha);\rho^\star_{t+h}\right)
\right],
\label{eq:pd_es_girf}
\end{equation}
with $m^{(\delta_g)}_{t+h}$ and $\rho^{\star(\delta_g)}_{t+h}$ defined as in
\eqref{eq:pd_es_coeffs} with $(\mu_{t+h},s_{t+h})$ replaced by
$(\mu^{(\delta_g)}_{t+h},s^{(\delta_g)}_{t+h})$.
\end{cor}

Appendix~\siref{app:proof_pd_es} gives the derivation. The argument $m_{t+h}$ is
the probit of the mean response \eqref{eq:GIRF_PD} and $\rho^\star_{t+h}$ is a
dispersion loading vanishing as $s_{t+h}\to0$, so the expected shortfall reuses
the conditional moments already computed at no additional cost. It satisfies
$\mathrm{ES}_\alpha(h)\ge\mathrm{PD}^{\mathrm{aR}}_\alpha(h)$, reduces to
$f(\mu_{t+h})$ as $s_{t+h}\to0$, and converges to the mean response as
$\alpha\to0$. Together, Propositions~\ref{prop:Moments-Z} and~\ref{prop:GIRF-PD}
and Corollaries~\ref{cor:pd_ar} and~\ref{cor:pd_es} characterize the closed-form
dynamic response of the mean, the quantiles, and the expected shortfall of
portfolio default probabilities, without simulation.

\subsection{The Merton--Vasicek ASRF}
\label{sec:mv_instance}

In credit-risk practice, the leading probit specification is the Merton–Vasicek ASRF model
\citep{merton1974pricing,vasicek2002distribution,Gordy2003}, which underlies the Basel IRB capital framework. For an asymptotically granular, homogeneous portfolio with through-the-cycle default probability $p\in(0,1)$ and asset correlation $\rho\in(0,1)$, the point-in-time conditional default probability is given by the map \eqref{eq:probit_map_general},
\begin{equation}
    \pi(Z_t)
    =\Phi\!\left(\frac{\Phi^{-1}(p)-\sqrt{\rho}\,Z_t}{\sqrt{1-\rho}}\right)
    =f(Z_t),
    \qquad
    a=\frac{\Phi^{-1}(p)}{\sqrt{1-\rho}},\quad
    b=-\frac{\sqrt{\rho}}{\sqrt{1-\rho}}<0,
    \label{eq:pd_conditional_z}
\end{equation}
so that all responses of Section~\ref{sec:irf_transmission} apply verbatim under
this $(a,b)$. We calibrate from aggregate default rates $d_t\in(0,1)$,
$t=T_0,\ldots,T_Y$: the through-the-cycle PD is the sample mean
$\widehat p=T_d^{-1}\sum_{t}d_t$, $T_d=T_Y-T_0+1$, and, interpreting
$d_t=\pi(Z_t)$, the factor is reconstructed by inverting
\eqref{eq:pd_conditional_z}, the general reconstruction
\eqref{eq:reconstruct_general}:\footnote{Default rates equal to zero or one must
be adjusted before applying the inverse normal transformation.}
\begin{equation}
    Z_t(\rho,\widehat p)
    =\frac{\Phi^{-1}(\widehat p)-\sqrt{1-\rho}\,\Phi^{-1}(d_t)}{\sqrt{\rho}}.
    \label{eq:determines_z_hist}
\end{equation}
The asset correlation $\widehat\rho$ is chosen so that
$Z_t(\widehat\rho,\widehat p)$ has unit sample variance, which standardizes the
latent factor and separates its scale from the satellite sensitivity
\eqref{eq:Z}.

\begin{remark}[Invariance to the calibration choice]
\label{rem:calibration_invariance}
For a fixed observed default-rate series $\{d_t\}$, the calibration pair
$(\widehat p,\widehat\rho)$ only rescales the reconstructed latent factor. The
reported conditional PD responses are therefore invariant to this choice once the
satellite is re-estimated consistently; see
Appendix~\siref{app:calibration_invariance}.
\end{remark}

\subsection{Uncertainty propagation}
\label{sec:posterior_implementation}

The results above are stated conditional on the parameters of the two blocks.
In the empirical implementation we estimate both blocks by Bayesian methods
and evaluate the closed-form expressions over posterior draws, so that the
reported responses inherit the estimation uncertainty of the macro-financial
transmission and of the credit-risk bridge.

The VAR is estimated under a conjugate Normal--Inverse--Wishart prior on
$(c,A_1,\ldots,A_P,\Sigma_u)$ \citep{kadiyala1997numerical}, following
standard practice for macroeconomic VARs \citep{banbura2010large}. Conditional
on the satellite regressors $Y_t^{(s)}$, the satellite equation is estimated
by conjugate Bayesian linear regression under the non-informative prior
$p(\beta_0,\beta,\sigma_\eta^2)\propto\sigma_\eta^{-2}$, which yields a
Normal--Inverse--Gamma posterior centered at the least-squares estimates
\citep{koop2003bayesian}. The Merton parameters $\widehat p$ and $\widehat\rho$ are calibrated as in
Section~\ref{sec:mv_instance} (sample-mean default rate and unit-variance
normalization), and the reconstructed
factor $Z_t$ is treated as observed in the satellite estimation.\footnote{By Remark~\ref{rem:calibration_invariance}, these calibration choices
only rescale the latent factor and therefore do not add uncertainty to the
reported conditional PD responses; see Appendix~\siref{app:calibration_invariance}.
The posterior bands are conditional on the observed default-rate series
$\{d_t\}$ and do not propagate sampling uncertainty. A binomial observation
equation could be used when default counts and exposures are available; in
Section~\ref{sec:empirical_analysis}, only the aggregate delinquency-rate ratio
is observed.}

The two posteriors are independent by construction, since the blocks are
estimated separately. We form joint draws by pairing independent draws
\[
    \theta_Y^{(b)}
    =
    \big(c^{(b)},A_1^{(b)},\ldots,A_P^{(b)},\Sigma_u^{(b)}\big),
    \qquad
    \theta_Z^{(b)}
    =
    \big(\beta_0^{(b)},\beta^{(b)},(\sigma_\eta^2)^{(b)}\big),
    \qquad b=1,\ldots,B.
\]
For each macro-financial draw we compute the moving-average coefficients
$\{\Psi_h^{(b)}\}_{h\ge0}$ and the response $\psi_Y^{g,(b)}(h)$; each paired
satellite draw then propagates it to the systematic factor and to default
probabilities through \eqref{eq:GIRF-Z}, Proposition~\ref{prop:Moments-Z}, and
Proposition~\ref{prop:GIRF-PD}. This yields a posterior sample
$\{\psi_{\pi(Z)}^{g,(b)}(h)\}_{b=1}^{B}$, $h=0,\ldots,H$, which we summarize
by pointwise posterior medians and credible bands, the standard reporting
convention in stress-testing applications; joint inference about the entire
response path is possible along the lines of \citet{InoueKilian2022}. The
posterior dispersion of $\psi_{\pi(Z)}^{g}(h)$ aggregates two distinct
sources of uncertainty: the propagation of the shock through the
macro-financial system, carried by the VAR posterior, and the mapping from
macro-financial conditions to credit risk, carried by the satellite
posterior.

\subsection{A numerical benchmark against forward simulation}
\label{sec:simulation_benchmark}

We benchmark the closed-form responses of
Section~\ref{sec:irf_transmission} against forward simulation of the system at
the point estimate. Both are built from the same representation
$Z_{t+h}=\mu_{t+h}+\sum_{q=0}^{h}\beta^\top G_{h,q}\,u_{t+q}+\eta_{t+h}$: the
simulation draws the innovations $\{u_{t+q}\}$ and $\eta_{t+h}$ forward,
conditions the impact innovation on $u_{g,t}=\delta_g$, and maps the simulated
factor through the Merton--Vasicek model, using common random numbers across
the baseline and shocked paths. Parameters are held fixed, so the
comparison measures the computing time and the sampling precision required to
recover each response object by simulation;
Appendix~\siref{app:simulation_design} details the design.

The posterior exercises reported below evaluate each response at
every posterior draw and, for the state-dependent results, at every historical
state. The closed form requires $B$ evaluations, each a small number of
univariate normal (and, for the expected shortfall, one bivariate normal)
distribution function calls, independent of the confidence level. A
simulation-based band nests an inner Monte Carlo of $N$ draws within each of
the $B\approx10^{4}$ posterior draws, an $O(B\times N)$ computation. With
$N=10^{6}$, reproducing the tail bands takes about twelve hours in our
implementation, against a few seconds for the closed form.

Reducing $N$ lowers this cost at the expense of precision in the
tail. At $N=10^{6}$, the Monte Carlo standard error of the peak response is
$3\times10^{-5}$ percentage points for the mean, $9\times10^{-4}$ for the
$99\%$ quantile, and $3\times10^{-3}$ for the $99.9\%$ quantile, about $6\%$
of the corresponding response; the error declines at the $1/\sqrt{N}$ rate
(Table~\ref{tab:sim_convergence}). At this rate, matching the precision of the
closed form in the deep tail requires on the order of tens of millions of
draws per evaluation.

The two estimates agree: across confidence levels, the difference
between the simulated and closed-form peak responses is of the order of the
Monte Carlo standard error, at most $2\times10^{-4}$ percentage points
(Figure~\ref{fig:simulation_benchmark}). Since both estimate the same
population object, this agreement is a numerical check of
Propositions~\ref{prop:Moments-Z} and~\ref{prop:GIRF-PD} and of their
implementation.

\begin{figure}[!htbp]
\centering
\includegraphics[width=0.62\linewidth]{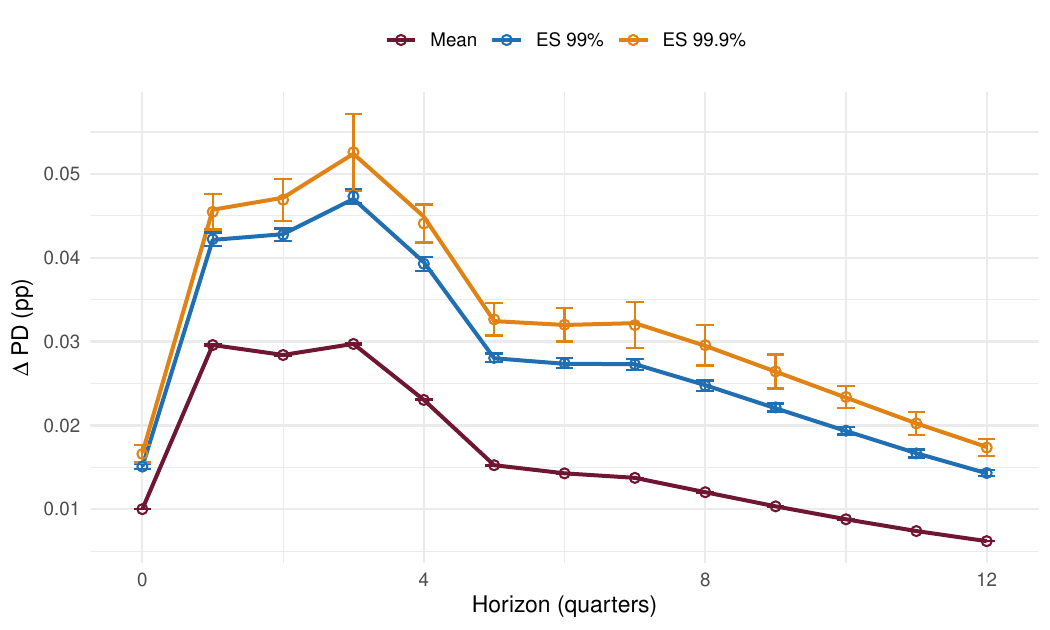}
\caption{\textbf{Closed-form responses versus forward simulation.}
\textit{Notes: Lines are the closed-form mean, $99\%$ and $99.9\%$
expected-shortfall responses to a one-standard-deviation GPR innovation at the
point estimate; points are the forward-simulation means over $100$ replications
($N=10^{6}$) with $\pm 2$ Monte Carlo standard-error bars. Both are built from
identical innovation draws with common random numbers across the baseline and
shocked paths.}}
\label{fig:simulation_benchmark}
\end{figure}

\begin{table}[!htbp]
\centering
  \scriptsize
  \tablestretch
  \compacttablecols
\caption{\textbf{Monte Carlo sampling error of the $99.9\%$ tail responses.}}
\label{tab:sim_convergence}
\begin{tabular}{lcccc}
\toprule
Draws $N$ & $10^3$ & $10^4$ & $10^5$ & $10^6$ \\
\midrule
RMSE, $99.9\%$ PD-aR (pp) & $0.051$ & $0.027$ & $0.0100$ & $0.0030$ \\
\quad in \% of the peak response & $101$ & $53$ & $20$ & $6$ \\
RMSE, $99.9\%$ ES (pp)    & $0.058$ & $0.020$ & $0.0061$ & $0.0017$ \\
\quad in \% of the peak response & $110$ & $37$ & $12$ & $3$ \\
\bottomrule
\end{tabular}%
\par\medskip
\begin{minipage}{0.84\textwidth}
\footnotesize
\textit{Notes: Root mean squared error of the forward-simulated peak $99.9\%$
response relative to the closed form, over $100$ replications per $N$, at the
point estimate. The error declines at approximately the $1/\sqrt{N}$ rate; the closed form is
exact. The closed-form peak responses are $0.050$ (PD-aR) and
$0.053$ (ES) percentage points.}
\end{minipage}
\end{table}

\section{Geopolitical stress testing of U.S. credit risk}
\label{sec:empirical_analysis}
This section applies the framework to U.S. credit risk.\footnote{Code and data:
\url{https://doi.org/10.5281/zenodo.21797928}.} The empirical exercise traces
the transmission of geopolitical-risk innovations from the macro-financial
system to the latent systematic credit factor and to the distribution of
portfolio default probabilities (PDs). After presenting the baseline data and
specification, we use the model for four exercises: a standard impulse-response
analysis, a historical-episode analysis that evaluates large GPR innovations at
their own initial conditions, a short-sample exercise for the credit-risk block,
and a comparison with the perfect-foresight convention used in supervisory
stress tests.\footnote{We report three robustness exercises in the appendix: a
direct test of the satellite exclusion restriction
(Appendix~\siref{app:orthogonality}, Table~\siref{tab:direct_channel_satellite});
the credit-risk response across alternative VAR information sets
(Appendix~\siref{app:alternative_information_sets},
Table~\siref{tab:dralacbn_information_set_robustness}); and BVAR stability checks
(Table~\siref{tab:var_stability_info_sets}). The qualitative conclusions are
unchanged throughout.}

\subsection{Data and empirical specification}
\label{sec:data_empirical_specification}

The application uses quarterly U.S. data from 1986:Q1 to 2024:Q4. The shock
variable is the Geopolitical Risk Index (GPR) of
\citet{caldara2022geopolitical}, aggregated to calendar-quarter frequency.
The index is well suited to the exercise because it provides a long,
event-based measure of geopolitical threats and acts, with peaks corresponding
to major geopolitical episodes (Figure~\ref{fig:GPR_index}).

The baseline credit-risk measure is the Delinquency Rate on All Loans and
Leases at all U.S. commercial banks (FRED: \texttt{DRALACBN}). This is a
delinquency rate (the share of loans past due), which is a broader measure than,
and not identical to, a regulatory (Article~178) default rate; we use it as a
long, publicly available proxy for realized portfolio-level default risk
because it spans several credit cycles and allows us to illustrate a key
motivation for the modular architecture: combining long macro-financial histories
with credit-risk histories that are shorter, coarser, or more limited in scope.
Its aggregate nature also implies that the estimated PD response is likely to be
attenuated, since the series pools heterogeneous borrowers, loan types, sectors,
regions, and bank exposures and therefore mainly reflects the systematic
component of default risk.

Over its available sample, 1986:Q1--2024:Q4, the delinquency rate averages $3.20\%$ and ranges from
$1.19\%$ to $7.50\%$, with the maximum reached during the 2009–2010 credit
downturn (Figure~\ref{fig:US_default}). The Merton–Vasicek inversion of
Section~\ref{sec:macro_credit_architecture} gives a through-the-cycle default
probability $\widehat p=3.20\%$ and an asset correlation
$\widehat\rho=0.051$. This correlation is low relative to the Basel IRB
benchmarks for corporate exposures (between $0.12$ and $0.24$), as expected for
a broadly diversified aggregate delinquency proxy. Since the gap between
quantile and mean responses widens with $\rho$, the tail amplification reported
below for this aggregate proxy is a lower bound on that of a less diversified
portfolio, or of a stressed regime in which correlation rises, rather than a
built-in margin of prudence.

\begin{figure}[!htbp]
\centering
\begin{subfigure}[t]{0.49\linewidth}
\centering
\includegraphics[width=\linewidth]{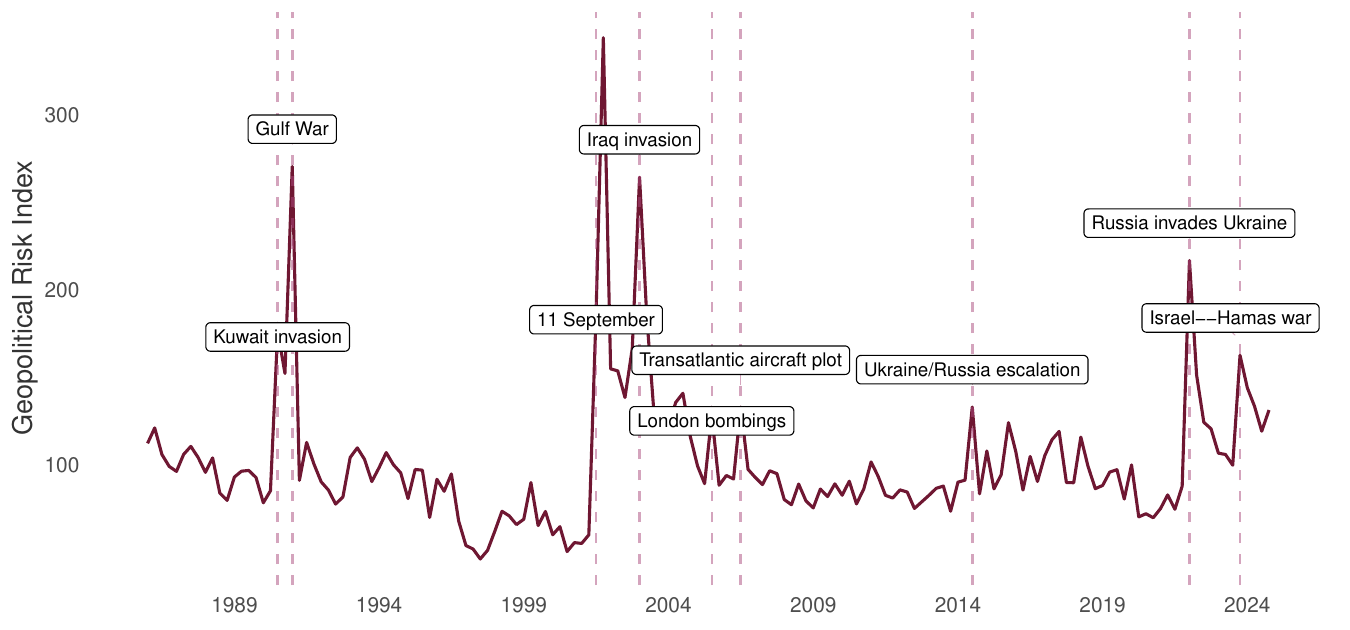}
\caption{Geopolitical Risk Index.}
\label{fig:GPR_index}
\end{subfigure}
\hfill
\begin{subfigure}[t]{0.49\linewidth}
\centering
\includegraphics[width=\linewidth]{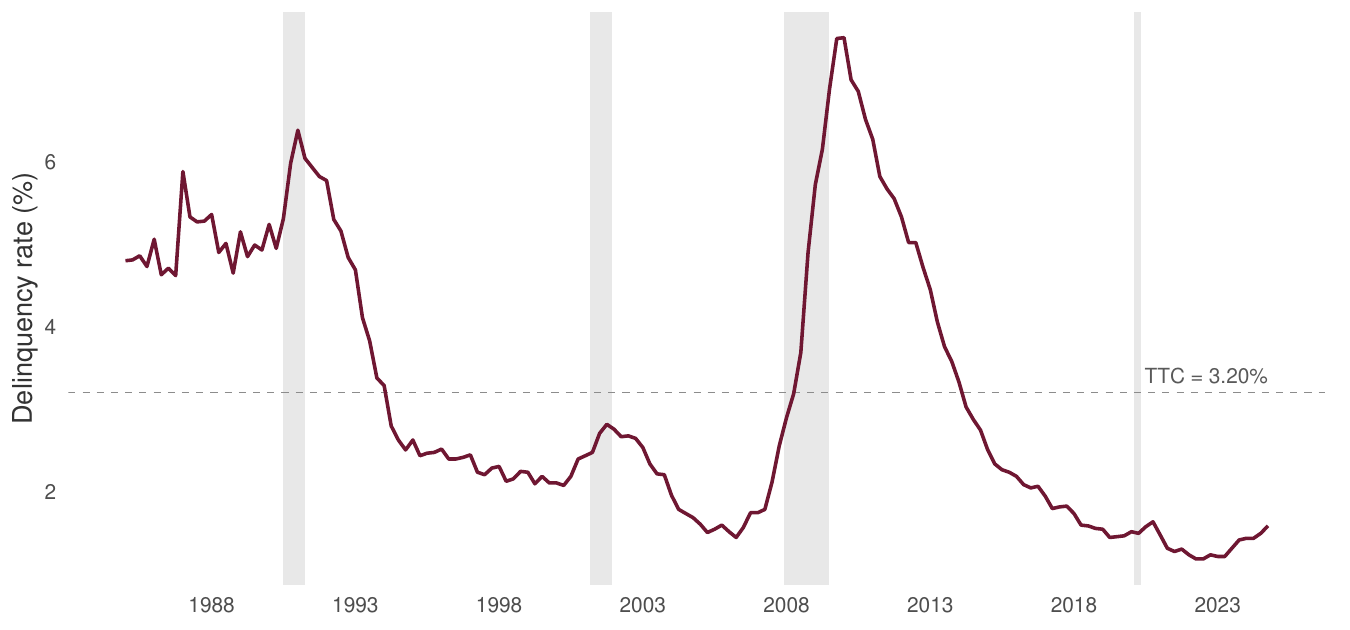}
\caption{U.S. aggregate delinquency rate.}
\label{fig:US_default}
\end{subfigure}
\caption{\textbf{Data: geopolitical risk and U.S. credit risk.}
\textit{Notes: Panel (a) plots quarterly averages of the Geopolitical Risk
Index of \citet{caldara2022geopolitical}, 1986:Q1--2024:Q4, whose peaks
correspond to major geopolitical episodes. Panel (b) plots the delinquency
rate on all loans and leases at all U.S. commercial banks (FRED
\texttt{DRALACBN}), quarterly{, 1986:Q1--2024:Q4,} with NBER recessions shaded; over the sample it
averages $3.20\%$ and ranges from $1.19\%$ to $7.50\%$, peaking in the
2009--2010 credit downturn. The delinquency rate is a broad proxy for realized
portfolio default risk and serves as the credit-risk series throughout.}}
\label{fig:data_series}
\end{figure}

The macro-financial VAR includes the GPR index and five real-side variables:
real private investment per capita, real GDP per capita, private employment
per capita, the real oil price, and year-on-year CPI inflation. The baseline
information set is intentionally parsimonious and focuses on the real-activity
channel through which geopolitical shocks can affect credit risk. Broader
information sets including monetary and uncertainty variables are considered in
Appendix~\siref{app:alternative_information_sets}, and variable definitions are
reported in Appendix~\siref{app:variables_detail}. The VAR is estimated under the
Normal–Inverse–Wishart prior described in
Section~\ref{sec:posterior_implementation}. For each posterior draw, we compute
the generalized impulse responses to a one-standard-deviation innovation in the
GPR equation and propagate them through the credit-risk block.

The satellite equation links the reconstructed factor $Z_t$ to current and
lagged non-GPR covariates from the baseline information set, with lags up to four
quarters. To avoid conditioning the transmission on a single specification, we
average across parsimonious satellite models using Schwarz weights
(Appendix~\siref{app:satellite_bma}). The GPR index is excluded from the satellite,
so geopolitical risk reaches credit only through the macro-financial block; this
exclusion restriction, orthogonality of the satellite error to the
macro-financial innovations, identifies the closed-form transmission.
Appendix~\siref{app:orthogonality} relaxes it with a control-function term and
derives the corresponding correction: the estimated direct channel is
statistically insignificant and leaves the peak PD response unchanged
(Tables~\siref{tab:direct_channel_satellite}–\siref{tab:direct_channel_pd}).

Table~\ref{tab:regZ} reports the model-averaged coefficients and posterior
inclusion probabilities for the most relevant satellite terms. The satellite is
estimated on $152$ quarterly observations. The systematic factor loads
positively on real investment, whose contemporaneous term has inclusion
probability one, and negatively on the real oil price, with weaker contributions
from inflation and GDP. This gives the transmission a simple economic
interpretation: geopolitical shocks that weaken real activity lower the
systematic credit factor and therefore increase portfolio PDs.

\begin{table}[!htbp]
  \centering
  \scriptsize
  \tablestretch
  \compacttablecols
  \caption{\textbf{Model-averaged satellite for the systematic factor $Z$.}}
  \label{tab:regZ}
    \begin{tabular}{lccc}
      \toprule
      & Post. mean & Post. SD & Incl. prob. \\
      \midrule
      Intercept
      & 12.827 & 0.636 & 1.00 \\
      $\log(\text{Investment pc})_{t}$
      & $0.039^{***}$ & 0.010 & 1.00 \\
      $\log(\text{Investment pc})_{t-2}$
      & $0.016^{***}$ & 0.014 & 0.64 \\
      $\log(\text{Oil real})_{t-4}$
      & $-0.004^{***}$ & 0.003 & 0.63 \\
      Inflation YoY$_{t-4}$
      & $0.059^{***}$ & 0.059 & 0.52 \\
      $\log(\text{Oil real})_{t}$
      & $-0.001^{***}$ & 0.002 & 0.39 \\
      $\log(\text{Oil real})_{t-1}$
      & $-0.002^{***}$ & 0.003 & 0.33 \\
      $\log(\text{GDP pc})_{t-4}$
      & $-0.015^{***}$ & 0.023 & 0.32 \\
      $\log(\text{GDP pc})_{t}$
      & $-0.017^{***}$ & 0.026 & 0.30 \\
      \midrule
      Observations
      & \multicolumn{3}{r}{152} \\
      $R^2$ 
      & \multicolumn{3}{r}{0.89} \\
      Out-of-sample $R^2$
      & \multicolumn{3}{r}{0.85} \\
      \bottomrule
    \end{tabular}%
  \par
  \medskip
  \begin{minipage}{0.72\textwidth}
  \footnotesize
  \textit{Notes: The table reports the posterior mean, posterior
  standard deviation, and inclusion probability of each satellite coefficient
  under Schwarz model averaging. The reported mean and standard deviation are
  moments of the model-averaging mixture and therefore include the zeros of the
  specifications that exclude the term. Stars $^{*}/^{**}/^{***}$ denote a
  posterior probability of the coefficient sign exceeding $0.95$, $0.975$, and
  $0.99$, respectively, \emph{conditional on inclusion} (evaluated on the draws
  in which the term is included); they are shown only for terms with inclusion
  probability at least $0.10$, below which the conditional sign probability is
  rebuilt from too few draws to be reliable. $R^2$ and out-of-sample $R^2$ are
  computed for the single best specification.
  }
\end{minipage}
\end{table}

\subsection{Dynamic responses}
\label{sec:empirical_results}

We propagate the one-standard-deviation innovation in the GPR equation through
the estimated system: first to the macro-financial variables, then to the
systematic factor, and finally to portfolio PDs. All responses are reported as
posterior medians with pointwise $68\%$ and $90\%$ credible bands.

\textit{Macro-financial responses.}
Figure~\ref{fig:irfs_var} reports the responses of the macro-financial
variables. The GPR innovation depresses real activity: investment, GDP per
capita, and private employment per capita fall and stay below baseline for
several quarters, while the real oil price and inflation respond more
transiently. The pattern is consistent with the contractionary effects of
geopolitical risk documented by \citet{caldara2022geopolitical}.

\begin{figure}[!htbp]
\centering
\includegraphics[width=0.85\linewidth]{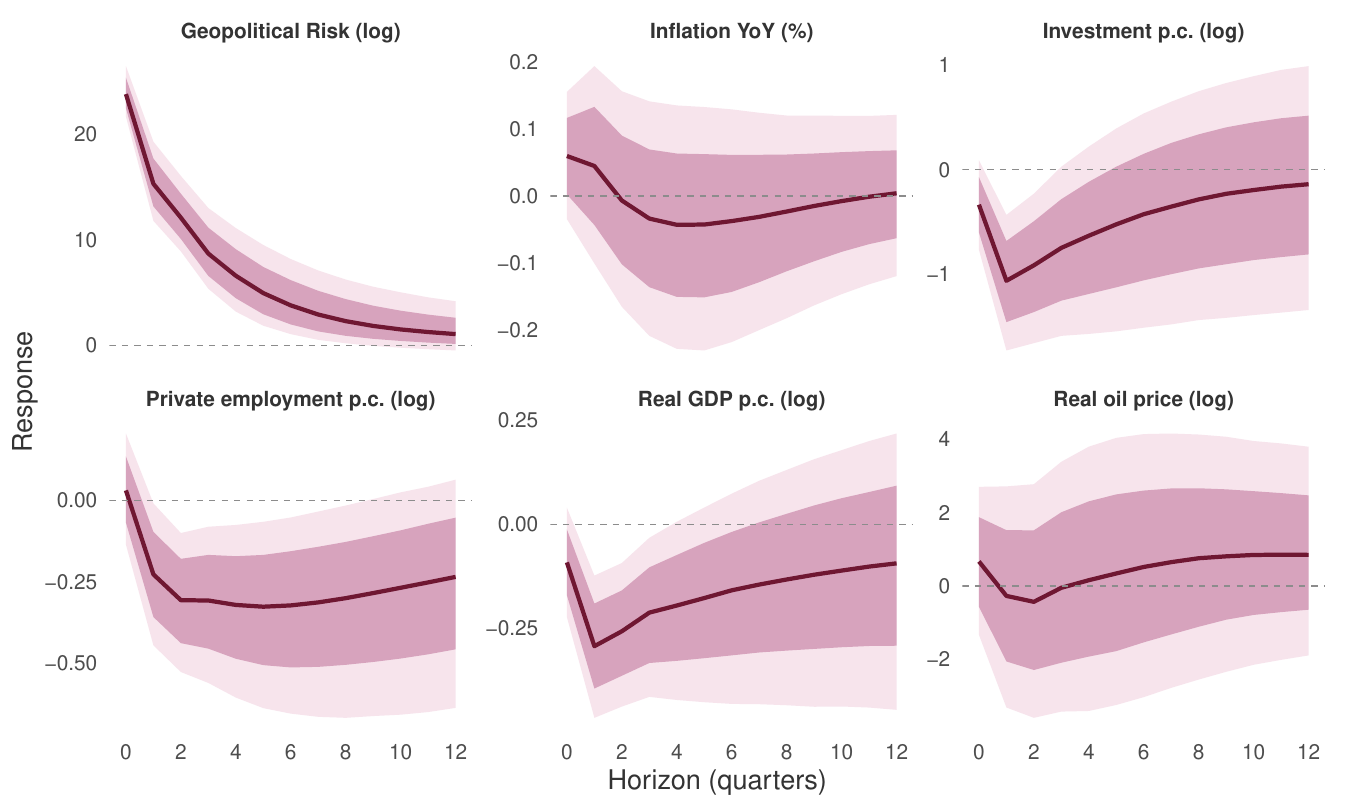}
\caption{\textbf{Macro-financial responses to a GPR innovation.}
\textit{Notes: Posterior medians (solid) with $68\%$ (dark) and $90\%$
(light) pointwise credible bands. Responses are generalized impulse
responses to a one-standard-deviation innovation in the GPR equation.}}
\label{fig:irfs_var}
\end{figure}

\medskip

\textit{Transmission to the systematic factor and portfolio PDs.}
The innovation first lowers the systematic credit factor $Z$: its response is
negligible on impact, turns negative within a few quarters, and reverts toward
baseline thereafter, combining the macro-financial propagation in the VAR with
the positive exposure of $Z$ to real activity in the satellite equation.
Through the Merton–Vasicek mapping, this movement translates into an increase
in portfolio PDs (Figure~\ref{fig:girf_Z_PD}). The median PD response peaks
after three quarters at $+0.033$ percentage points, about $1\%$ of the
through-the-cycle level; the posterior probability of a positive response at
the peak is $0.94$. The median response cumulates to about $0.26$
percentage-point-quarters of additional default risk over the twelve-quarter
horizon.

\begin{figure}[!htbp]
\centering
\includegraphics[width=0.60\linewidth]{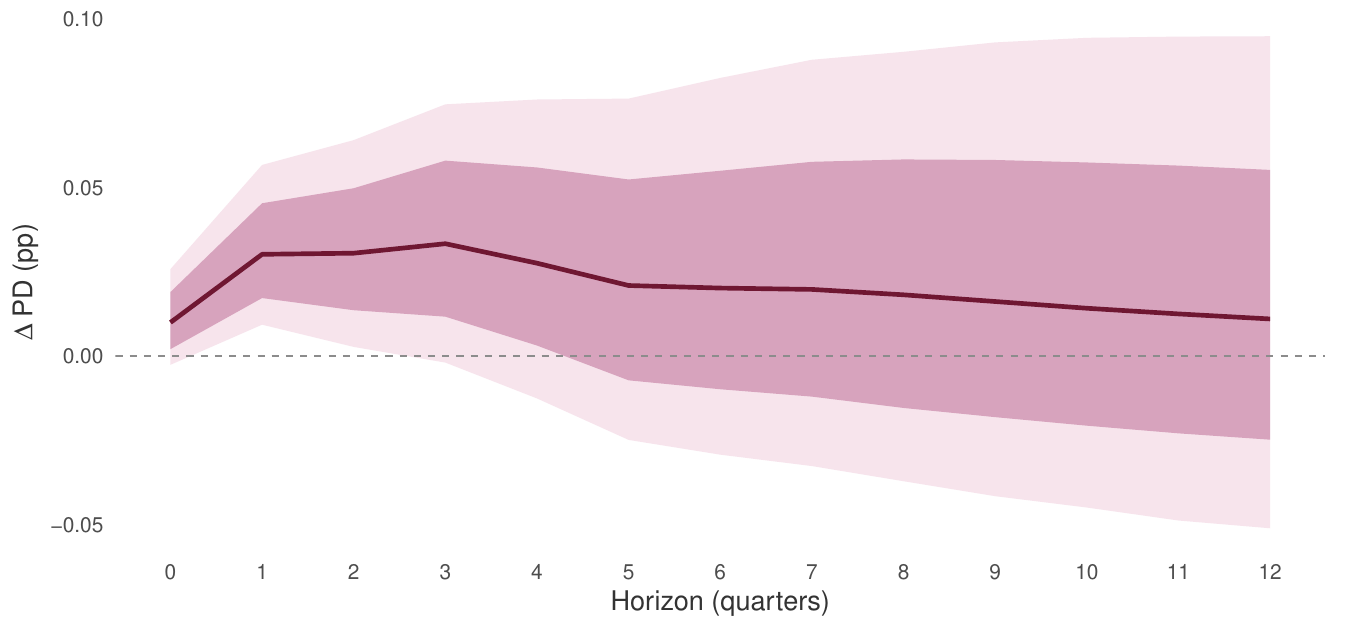}
\caption{\textbf{Generalized response of portfolio default probabilities to a
GPR innovation.}
\textit{Notes: Generalized response of portfolio default probabilities (in
percentage points) to a one-standard-deviation innovation in the GPR equation,
obtained by mapping the response of the systematic credit factor $Z$ through
the Merton–Vasicek model with $\widehat p=3.20\%$ and $\widehat\rho=0.051$.
Posterior median (solid) with $68\%$ (dark) and $90\%$ (light) pointwise
credible bands.}}
\label{fig:girf_Z_PD}
\end{figure}

\textit{Tail responses.}
Capital and provisioning depend on the upper tail of the conditional
default-probability distribution rather than on its center.
Corollary~\ref{cor:pd_ar} delivers this tail response in closed form.
Panel (a) of Figure~\ref{fig:girf_pd_ar} reports the generalized response of the
portfolio PD-at-Risk (PD-aR). Its timing matches the mean response, with a peak at the
third quarter, but its magnitude is larger: the $99\%$ PD-aR rises by
$+0.050$ percentage points at the peak and the $99.9\%$ PD-aR by $+0.056$,
against $+0.033$ for the mean, a ratio of about $1.5$ at the $99\%$ level and
$1.7$ at the $99.9\%$ level. This amplification
follows from the curvature of the Merton–Vasicek mapping, under which a given
decline in the systematic factor shifts the upper tail of the conditional PD
distribution more than its center. Measured against the through-the-cycle
PD-aR levels of $8.6\%$ ($99\%$) and $11.8\%$ ($99.9\%$), the response is
modest. The tail deteriorates more than the center in absolute (percentage-point)
terms, but less in proportional terms ($0.050/8.6\approx0.6\%$ at the $99\%$
level against $0.033/3.20\approx1.0\%$ for the mean). The larger tail effect is
thus one of absolute magnitude, the quantity relevant for capital and
provisions. The expected shortfall, the coherent tail measure of
Corollary~\ref{cor:pd_es} shown in panel (b) of
Figure~\ref{fig:girf_pd_ar}, delivers the same message: the $99\%$ and
$99.9\%$ ES responses peak at $+0.053$ and $+0.059$ percentage points, each
slightly above the corresponding PD-aR as $\mathrm{ES}_\alpha\ge
\mathrm{PD}^{\mathrm{aR}}_\alpha$ requires, with the same three-quarter timing.

\medskip
\textit{A modest aggregate response.}
The moderate mean response reflects the credit-risk measure and the shock. The
delinquency rate on all loans and leases is a broad aggregate proxy; in a
granular Merton--Vasicek--Gordy portfolio idiosyncratic risk is diversified
away, so the aggregate PD responds only to the systematic component, much as
aggregate export volumes respond modestly to large exchange-rate movements in
\citet{berman2012different}. The interpretation is also specific to the GPR
measure of \citet{caldara2022geopolitical}, which captures adverse military and
security risks rather than trade fragmentation, sanctions, or technology
restrictions. A modest aggregate mean effect thus coexists with
larger tail responses and substantial state dependence.

\begin{figure}[!htbp]
  \centering
  \begin{subfigure}[t]{0.49\linewidth}
    \centering
    \includegraphics[width=\linewidth]{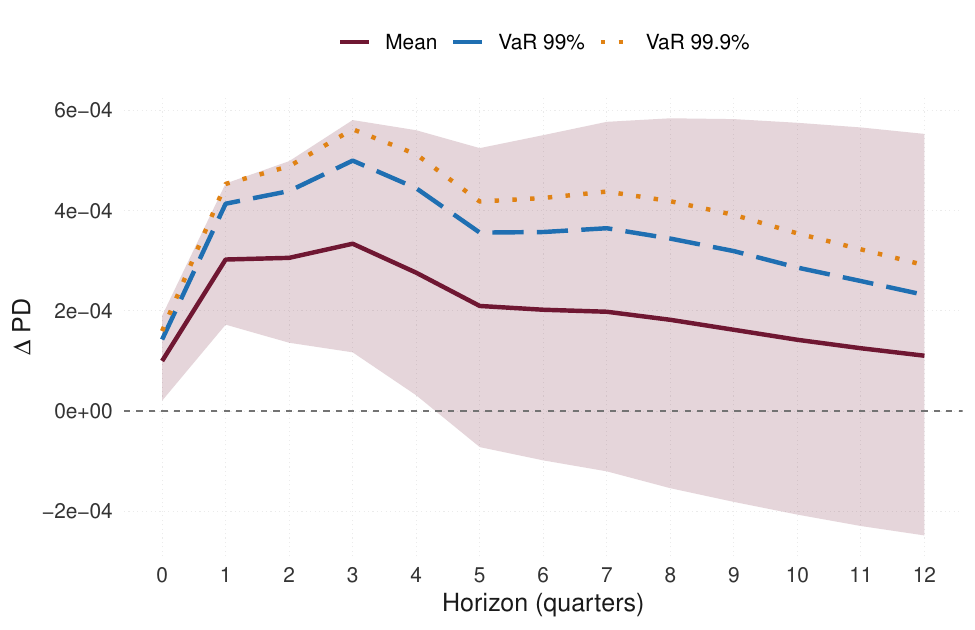}
    \caption{PD-at-Risk (Corollary~\ref{cor:pd_ar}).}
    \label{fig:girf_pd_ar_panel}
  \end{subfigure}
  \hfill
  \begin{subfigure}[t]{0.49\linewidth}
    \centering
    \includegraphics[width=\linewidth]{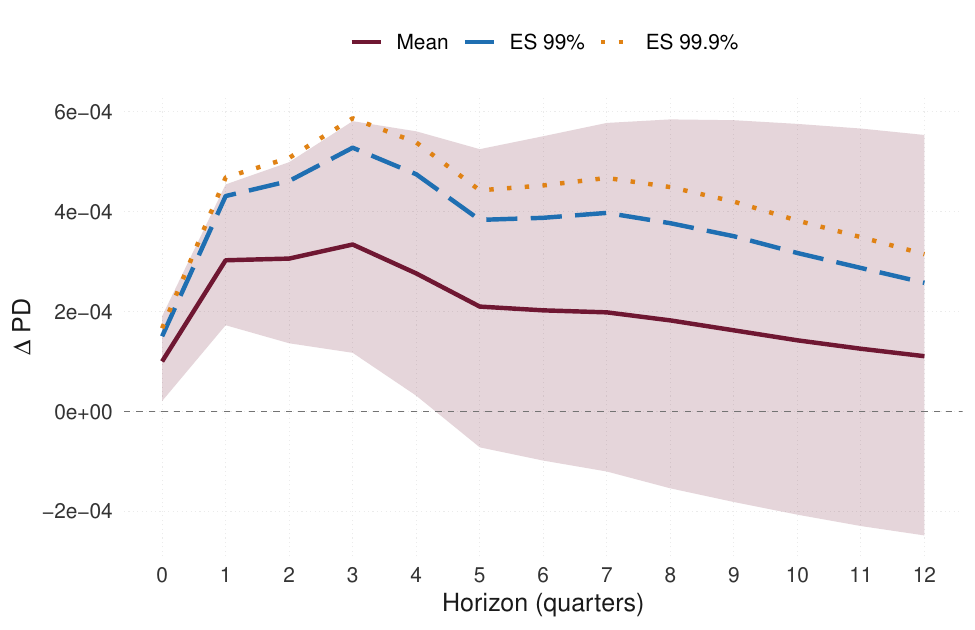}
    \caption{Expected shortfall (Corollary~\ref{cor:pd_es}).}
    \label{fig:girf_pd_es_panel}
  \end{subfigure}
  \caption{\textbf{Tail responses of portfolio default probabilities to a GPR
  innovation.}
  \textit{Notes: Posterior-median generalized responses (percentage points) of
  the mean portfolio PD and of its $99\%$ and $99.9\%$ tail measures to a
  one-standard-deviation GPR innovation, obtained by mapping the
  systematic-factor response through the Merton–Vasicek model
  ($\widehat p=3.20\%$, $\widehat\rho=0.051$). Panel (a) reports the PD-at-Risk
  and panel (b) the expected shortfall. In both panels the mean response is shown with its $68\%$
  pointwise credible band. Both tail measures peak at the third quarter, like
  the mean, but are larger, reflecting the curvature of the map that shifts the
  upper tail of the conditional PD distribution more than its center; the
  expected shortfall lies above the corresponding PD-at-Risk at each level, as a
  coherent measure requires.}}
  \label{fig:girf_pd_ar}
\end{figure}

\subsection{Geopolitical episodes and state dependence}
\label{sec:scenario_2001}
Stress-testing frameworks require severe yet plausible scenarios. The
estimated VAR provides a direct device for constructing them: its history of
standardized reduced-form GPR innovations identifies large historical shocks
in the units of the model (Figure~\ref{fig:gpr_innov}). The three largest
positive innovations in the sample associated with distinct events correspond
to recognizable episodes: the
invasion of Kuwait in 1990:Q3 ($3.14$ standard deviations), the September 11
attacks in 2001:Q3 ($4.45$), and the invasion of Ukraine in 2022:Q1
($3.54$).\footnote{The 2001:Q4 innovation ($3.90$ standard deviations) is
larger than the Gulf War one but continues the September 11 episode; we retain
one innovation per event.} Because the response of Proposition~\ref{prop:GIRF-PD} is
history-dependent through the conditional mean of the systematic factor, each
episode can be evaluated at its own initial conditions, and these episodes
happen to span the credit cycle: the Gulf War struck a stressed credit system
(model-implied PD of $4.6\%$ at the shock date), the September 11 attacks an
intermediate one ($2.0\%$), and the Ukraine invasion a benign one ($1.4\%$).
The macro-financial propagation is identical across these experiments, since
the VAR is linear; the comparison isolates the position of the credit cycle
in the Merton--Vasicek map. Parameters are full-sample estimates, so the
exercise is a counterfactual within the estimated model rather than a
real-time analysis.
\begin{figure}[!htbp]
    \centering
    \includegraphics[width=0.70\linewidth]{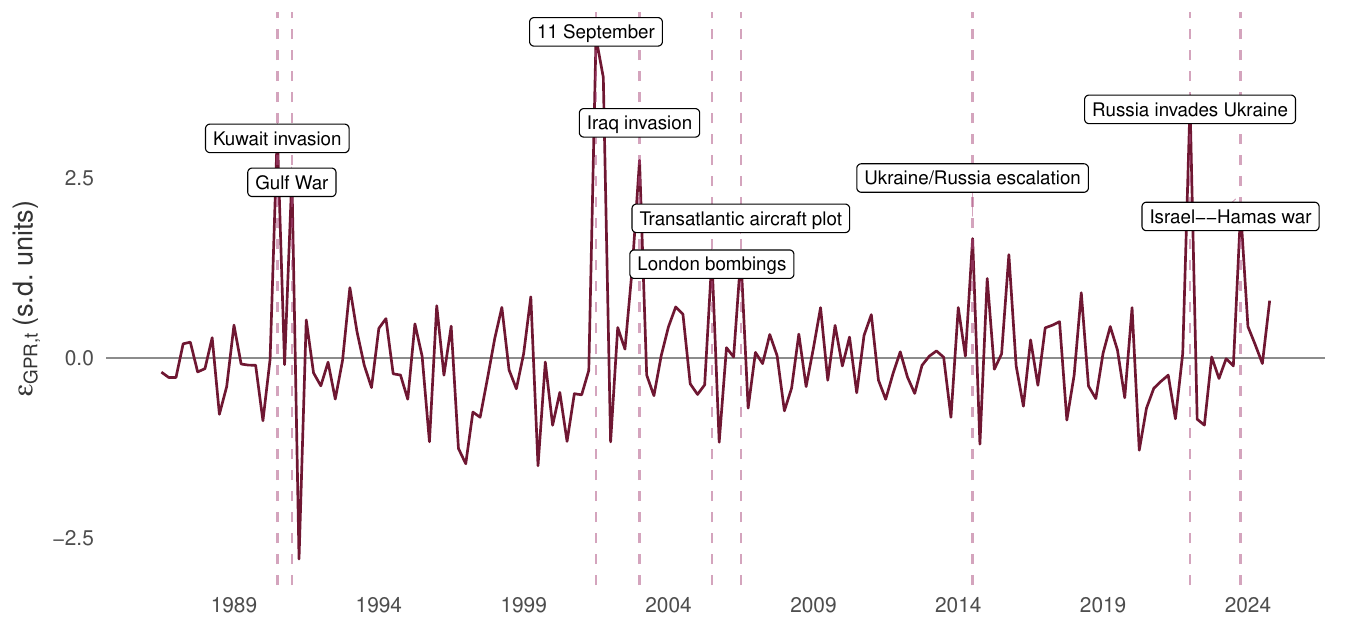}
    \caption{\textbf{Standardized GPR innovations.}
    \textit{Notes: Standardized reduced-form innovations in the GPR equation,
    $\varepsilon_t^{g}=u_{gt}/\sqrt{\sigma_{gg}}$, where $g$ indexes the GPR
    equation, $u_{gt}$ is its reduced-form innovation in the Bayesian VAR, and
    $\sigma_{gg}$ is the corresponding innovation variance, evaluated at
    posterior-median parameters. Values are in standard-deviation units; positive values
    correspond to unexpected increases in geopolitical risk.}}
    \label{fig:gpr_innov}
\end{figure}
Table~\ref{tab:state_dependence} reports the responses. Evaluated
at its own date, the Gulf War innovation produces a larger PD response
($+0.278$ percentage points at the peak) than the September 11 innovation
($+0.219$), although the latter shock is more than forty percent larger. A smaller geopolitical shock
generated a larger credit-risk response because it struck a weaker credit
system. The marginal credible bands of the two responses overlap
(Table~\ref{tab:state_dependence} reports $68\%$ intervals, $[0.099,0.477]$ for
the Gulf War against $[0.077,0.376]$ for September 11), but the comparison is
sharper than this overlap suggests: the two experiments are evaluated on
identical posterior draws, so estimation uncertainty is common to both, and the
posterior probability that the Gulf War response exceeds the September 11
response at the peak is $0.94$. Beyond the three episodes, the peak response to a
common one-standard-deviation innovation rises strongly with the baseline
expected PD across all admissible histories (at posterior-median parameters;
the rank correlation between the two is $0.98$), from a minimum of $0.029$
percentage points (2021:Q1) to a maximum of $0.131$ at the most stressed state
(2009:Q4), a factor of $4.6$ over the historically observed range of credit
conditions.

These results characterize the state dependence of the macro-to-credit
transmission. The state of the credit cycle is the operative margin: the same
one-standard-deviation GPR innovation generates a larger PD response when
initial credit conditions are weaker. By contrast, at a given initial state,
the size of the shock scales the response almost proportionally. From the
end-of-sample state, for instance, a September-11-sized shock yields
$+0.155$ percentage points, which is $4.6$ times the one-standard-deviation
response for a $4.45$-standard-deviation shock. Because the macro-to-factor
response $\psi_Z$ is state-independent (the VAR is linear), the posterior
probability that a GPR innovation raises the PD is the same across states at any
given horizon, equal to $0.94$ at the three-quarter peak; this is the same
quantity as the pairwise-dominance probability reported above for the Gulf
War/September~11 comparison, both being the posterior probability that $\psi_Z$
is negative at the peak. The peak occurs at three quarters for the stressed and
intermediate states, marginally earlier for the most benign histories where the
response is smallest. State dependence
therefore operates through the Merton--Vasicek credit mapping, rather than
through the macro-financial dynamics.
The same evidence rationalizes the two-block architecture. A linear VAR that
included the default rate directly would imply state-independent responses, a
flat relation between the baseline expected PD and the peak response, and
proportionality between the Gulf War and September 11 responses. Instead, the
estimated relation has an economically large positive slope, though close to
linear over the observed range, and the Gulf War/September 11 comparison
displays an inversion.
\begin{table}[!htbp]
  \centering
  \scriptsize
  \tablestretch
  \compacttablecols
  \caption{\textbf{Geopolitical episodes evaluated at their own initial conditions.}}
  \label{tab:state_dependence}
    \begin{tabular}{lccccc}
      \toprule
      Episode
      & \makecell{PD at shock\\date (pp)}
      & \makecell{Shock\\(s.d.)}
      & \makecell{Peak $\Delta$PD,\\own shock (pp)}
      & \makecell{Peak $\Delta$PD,\\1 s.d. (pp)}
      & Ratio \\
      \midrule
      Gulf War (1990:Q3)
      & 4.59 & 3.14 & \makecell{$0.278$\\{\scriptsize $[0.099,\,0.477]$}} & 0.087 & 2.59 \\
      September 11 (2001:Q3)
      & 1.97 & 4.45 & \makecell{$0.219$\\{\scriptsize $[0.077,\,0.376]$}} & 0.047 & 1.41 \\
      Ukraine invasion (2022:Q1)
      & 1.36 & 3.54 & \makecell{$0.106$\\{\scriptsize $[0.059,\,0.162]$}} & 0.029 & 0.87 \\
      End of sample (2025:Q1)
      & 1.45 & -- & -- & 0.033 & 1.00 \\
      \bottomrule
    \end{tabular}%
  \par
  \medskip
  \begin{minipage}{\textwidth}
  \footnotesize
  \textit{Notes: Generalized responses of the portfolio PD are
  evaluated conditional on the history preceding each episode, on identical
  posterior draws. ``PD at shock date'' is the posterior median of the
  model-implied baseline PD at $h=0$. ``Own shock'' replays the episode's
  standardized reduced-form GPR innovation; brackets report $68\%$ credible
  intervals. ``1 s.d.'' applies a common one-standard-deviation innovation to
  isolate state dependence. ``Ratio'' is the corresponding peak relative to the
  end-of-sample state. All peaks occur at $h=3$, except for the
  Ukraine episode, whose peak occurs at $h=1$. The episodes are the three
  largest positive GPR innovations in the sample associated with distinct
  events. Parameters are full-sample
  estimates.
  }
\end{minipage}
\end{table}
\subsection{Using shorter default histories}
\label{sec:short_default_histories}
The modular architecture is designed for the common case in which
macro-financial variables span long samples but default histories are shorter or
heterogeneous across definitions, regulation, and portfolio composition. A fully
joint model would either discard macro-financial history or impose a homogeneous
default series throughout; the proposed framework instead estimates the VAR on
the full macro sample and re-estimates only the credit-risk bridge on the
available default window. We keep the VAR on 1986:Q1--2024:Q4 and re-estimate the
Merton--Vasicek inversion and the satellite on default samples starting in
2005:Q1 and in 2015:Q1 (Figure~\ref{fig:short_default_windows}), applying the
same one-standard-deviation GPR innovation throughout. The response remains
positive in both windows, but its magnitude and uncertainty depend on the default
sample. The 2005 window contains the global financial crisis, a stressed episode
from which the satellite identifies the macro-credit sensitivity, whereas the
shorter 2015 window covers a mostly benign period, which weakens identification
and widens the bands. The long macro sample still identifies and propagates the
shock, while the credit-risk block adapts to the available default history.
\begin{figure}[!htbp]
    \centering
    \begin{subfigure}[t]{0.48\textwidth}
        \centering
        \includegraphics[width=\textwidth]{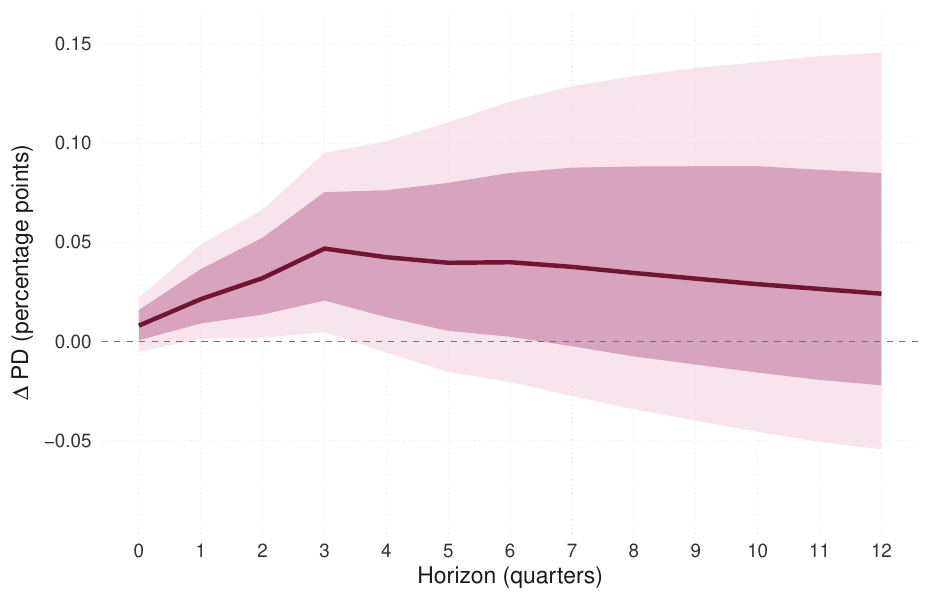}
        \caption{Default sample starts in 2005:Q1.}
        \label{fig:short_default_window_2005}
    \end{subfigure}
    \hfill
    \begin{subfigure}[t]{0.48\textwidth}
        \centering
        \includegraphics[width=\textwidth]{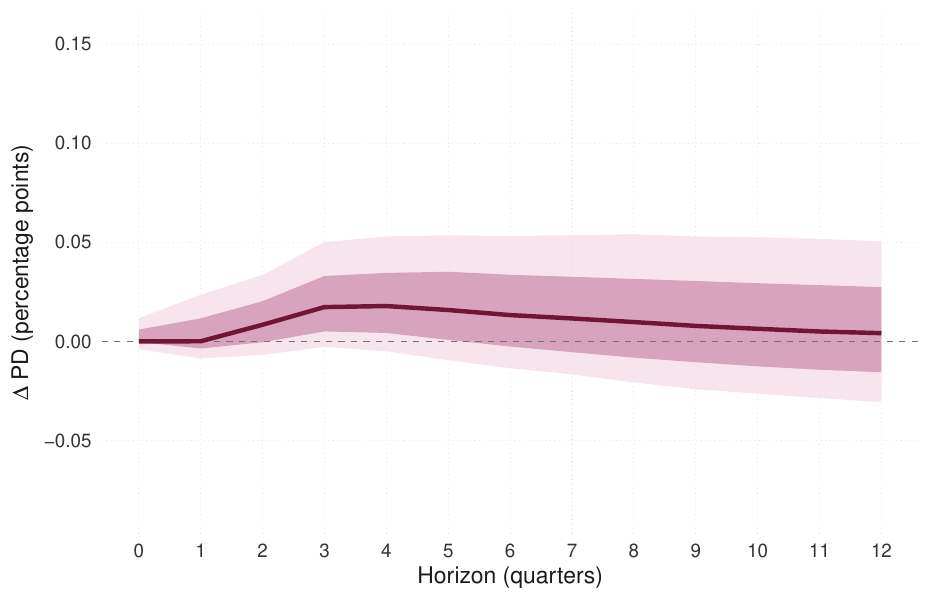}
        \caption{Default sample starts in 2015:Q1.}
        \label{fig:short_default_window_2015}
    \end{subfigure}
    \caption{\textbf{Portfolio-PD responses with shorter default histories.}
    \textit{Notes: The macro-financial VAR is estimated on the full
    1986:Q1--2024:Q4 sample, while the Merton--Vasicek credit-risk block and the
    satellite equation are re-estimated on shorter default samples starting in
    2005:Q1 and 2015:Q1. Each panel reports the portfolio-PD response to a
    one-standard-deviation GPR innovation; conventions as in
    Figure~\ref{fig:girf_pd_ar}.}}
    \label{fig:short_default_windows}
\end{figure}

\subsection{The cost of perfect foresight}
\label{sec:perfect_foresight}
The closed-form expression lets us isolate the effect of predictive
uncertainty from estimation uncertainty. The distinction matters because
supervisory stress tests treat the macroeconomic scenario as a known path:
banks ``should assume the subsequent path of a variable to be known in line
with the scenario,'' and provisions are projected under one path per scenario
\citep[Box~1, paragraphs~130 and~136]{EBA2025method}. In our framework this
convention amounts to suppressing part of the conditional variance that enters
the Merton--Vasicek map.

To make this precise, recall that the projected PD at horizon $h$ is the
nonlinear map $\pi(\cdot)$ applied to the latent factor
$Z_{t+h}=\beta_0+\beta^\top Y^{(s)}_{t+h}+\eta_{t+h}$ of equation~\eqref{eq:Z}.
Writing $\bar Y^{(s)}_{t+h}=\mathbb E[Y^{(s)}_{t+h}\mid\Omega_t]$ for the
conditional-mean scenario path, we evaluate three variants on identical
posterior draws:
\begin{align}
  \mathrm{PD}^{\mathrm{exact}}_h
    &= \mathbb E\!\left[\pi(Z_{t+h})\mid\Omega_t\right],
    \label{eq:pd_exact}\\[2pt]
  \mathrm{PD}^{\mathrm{PF}}_h
    &= \mathbb E_{\eta}\!\left[\pi\!\big(\beta_0+\beta^\top\bar Y^{(s)}_{t+h}
       +\eta_{t+h}\big)\right],
    \label{eq:pd_pf}\\[2pt]
  \mathrm{PD}^{\mathrm{plug\text{-}in}}_h
    &= \pi\!\left(\mathbb E[Z_{t+h}\mid\Omega_t]\right).
    \label{eq:pd_plugin}
\end{align}
The \emph{exact} variant integrates over the full predictive distribution of
$Z_{t+h}$, whose variance combines the VAR predictive variance of the
macro-financial path and the satellite variance $\sigma_\eta^2$. The
\emph{perfect-foresight} variant fixes the scenario path at its conditional
mean $\bar Y^{(s)}_{t+h}$ but keeps the satellite error $\eta_{t+h}$; it is the
closest counterpart to the supervisory convention. The \emph{plug-in} variant additionally drops $\eta_{t+h}$ and evaluates the
credit-risk map at the conditional mean of the latent factor,
$\pi(\mathbb E[Z_{t+h}\mid\Omega_t])$.\footnote{By
Corollary~\ref{cor:pd_ar}, the plug-in value is the median PD-aR, obtained at
$\alpha=0.5$.} Because $\pi$ is convex in the relevant region, Jensen's inequality
orders the three,
\begin{equation}
  \mathrm{PD}^{\mathrm{plug\text{-}in}}_h
  \;\le\; \mathrm{PD}^{\mathrm{PF}}_h
  \;\le\; \mathrm{PD}^{\mathrm{exact}}_h ,
  \label{eq:pf_jensen}
\end{equation}
so both conventions understate the expected PD, and the gap widens with the
predictive variance, hence with the horizon and in adverse states.

Table~\ref{tab:pf_bias} confirms that perfect foresight mainly affects
projected PD levels, not impulse responses. The perfect-foresight variant
understates the exact PD path by $2.1\%$ at the one-year horizon and by
$6.7\%$ at the three-year horizon; the plug-in approximation produces larger
biases, $3.9\%$ and $8.4\%$. By contrast, impulse responses are almost
unchanged: the convexity correction enters both the shocked and baseline
evaluations and largely cancels in the difference, leaving peak-response biases
of at most about $2\%$ for both the one-standard-deviation and the
September-11-sized shock. Perfect foresight therefore has limited effects on the
estimated sensitivity of credit risk to shocks, but it lowers projected
default-risk levels, which are directly relevant for provisioning and capital
projections.

A stylized provisioning calculation makes the level bias concrete. With a
loss-given-default of $\mathrm{LGD}=45\%$ and expected credit loss
$\mathrm{ECL}_h=\mathrm{EAD}\cdot\mathrm{LGD}\cdot\mathrm{PD}_h$, the ECL is
proportional to the projected PD, so the percentage biases of
Table~\ref{tab:pf_bias} carry over to provisions. At the three-year horizon the
exact ECL is $67.6$ basis points of EAD, against $63.0$ bp under perfect
foresight and $62.0$ bp under plug-in, an understatement of $4.6$ and $5.7$ bp
(about \$$0.46$m and \$$0.57$m per \$$1$bn of exposure); at the one-year horizon
the gap is $1.3$ to $2.5$ bp. The bias is small but systematic and one-signed,
and it grows with the horizon and in adverse states.

Relative to the plug-in benchmark, the gain from the closed form
is small for impulse responses (at most $2\%$ at the peak, since the convexity
correction enters the shocked and baseline evaluations symmetrically and
largely cancels) and larger for projected levels and the provisions based on
them ($6$ to $8\%$ at the three-year horizon). The plug-in evaluation also
returns no measure of dispersion, so the quantile and expected-shortfall
responses of Corollaries~\ref{cor:pd_ar} and~\ref{cor:pd_es} have no plug-in
counterpart.
\begin{table}[!htbp]
  \centering
  \scriptsize
  \tablestretch
  \compacttablecols
  \caption{\textbf{The cost of perfect foresight: PD levels and responses.}}
  \label{tab:pf_bias}
    \begin{tabular}{llccccc}
      \toprule
      & & Exact & PF & Plug-in & Bias PF (\%) & Bias plug-in (\%) \\
      \midrule
      PD level (pp) & $h=4$
      & 1.427 & 1.397 & 1.371 & $-2.1$ & $-3.9$ \\
      PD level (pp) & $h=8$
      & 1.458 & 1.390 & 1.366 & $-4.7$ & $-6.3$ \\
      PD level (pp) & $h=12$
      & 1.503 & 1.401 & 1.377 & $-6.7$ & $-8.4$ \\
      \midrule
      Peak $\Delta$PD (pp), one s.d. & $h=3$
      & 0.033 & 0.033 & 0.033 & $-0.1$ & $-1.2$ \\
      Peak $\Delta$PD (pp), 2001:Q3 & $h=3$
      & 0.155 & 0.154 & 0.152 & $-0.8$ & $-2.0$ \\
      \bottomrule
    \end{tabular}%
  \par
  \medskip
  \begin{minipage}{0.96\textwidth}
  \footnotesize
  \textit{Notes: The table compares three evaluations of the projected
  portfolio PD on identical posterior draws, differing only in the predictive
  variance. ``Exact'' integrates over the full predictive distribution of the
  latent factor (Proposition~\ref{prop:GIRF-PD}); ``PF'' (perfect foresight)
  fixes the macro-financial path at its conditional mean but keeps the
  satellite-error variance, the closest counterpart to supervisory stress tests
  \citep{EBA2025method}; ``Plug-in'' additionally drops the satellite error. The
  upper block reports PD levels (pp) at the one-, two-, and three-year horizons
  ($h=4,8,12$); the lower block reports peak PD responses (pp) to a
  one-standard-deviation GPR innovation and to a September-11-sized ($4.45$ s.d.)
  shock at the end-of-sample state ($0.155$ pp, distinct from the $+0.219$ pp at
  its own 2001:Q3 conditions in Section~\ref{sec:scenario_2001}). Biases are
  relative to the exact variant; negative values indicate understatement.
  }
\end{minipage}
\end{table}

\section{Conclusion}
\label{sec:conclusion}
This paper develops a closed-form impulse-response approach to credit-portfolio stress testing. The framework links a macro-financial VAR and a Gaussian satellite to the Merton--Vasicek ASRF map of bank capital regulation; the closed-form results hold for any probit-Gaussian observation model, recession and qualitative-outcome probits among them. Its contribution is to derive the exact response of the mean, the quantiles, and the expected shortfall of portfolio default probabilities to a macro-financial innovation, rather than evaluating the nonlinear map along a deterministic conditional-mean path. The method preserves the modular structure of internal and supervisory stress-testing systems, while allowing the macro-financial and credit-risk blocks to be estimated on different samples and their estimation uncertainty to be propagated through posterior simulation.

In the application to U.S. geopolitical risk, a one-standard-deviation innovation to the Geopolitical Risk Index raises mean portfolio default probabilities modestly, with a peak response of 0.033 percentage points after three quarters, but has larger effects on upper-tail default-probability quantiles. Historical episodes also show that the same type of shock generates larger credit-risk responses when initial credit conditions are weaker. A comparison with the perfect-foresight convention of supervisory stress tests further shows that treating the macroeconomic scenario as a known path leaves impulse responses almost unchanged but understates projected PD \emph{levels} at the three-year horizon (by about seven percent under perfect foresight, and up to about eight percent under the stricter plug-in approximation) because of the convexity of the Merton--Vasicek map in the empirically relevant low-PD region. These results indicate that geopolitical stress testing should account for both the distribution of the latent credit factor and the state of the credit cycle.

Two limitations qualify the empirical application. The credit-risk block uses an aggregate delinquency rate as a proxy for portfolio default; it is broader than a regulatory default rate and, being pooled across borrowers, loan types, sectors, and banks, likely attenuates the estimated responses. The framework is also deliberately top-down and does not capture feedback from credit conditions back to the macroeconomy. Applying the method to a disaggregated bank portfolio with regulatory default probabilities, which the modular design directly accommodates, is a natural next step.


\label{bib}

\setlength{\bibsep}{8pt}
\bibliography{bibliography.bib}

@article{acerbi2002coherence,
  author  = {Acerbi, Carlo and Tasche, Dirk},
  title   = {On the Coherence of Expected Shortfall},
  journal = {Journal of Banking \& Finance},
  year    = {2002},
  volume  = {26},
  number  = {7},
  pages   = {1487--1503},
  doi     = {10.1016/S0378-4266(02)00283-2}
}

@article{adrian2019vulnerable,
  author  = {Adrian, Tobias and Boyarchenko, Nina and Giannone, Domenico},
  title   = {Vulnerable Growth},
  journal = {American Economic Review},
  year    = {2019},
  volume  = {109},
  number  = {4},
  pages   = {1263--1289}
}

@article{albert1993bayesian,
  author  = {Albert, James H. and Chib, Siddhartha},
  title   = {{Bayesian} Analysis of Binary and Polychotomous Response Data},
  journal = {Journal of the American Statistical Association},
  year    = {1993},
  volume  = {88},
  number  = {422},
  pages   = {669--679},
  doi     = {10.1080/01621459.1993.10476321}
}

@article{banbura2010large,
  author  = {Ba\'{n}bura, Marta and Giannone, Domenico and Reichlin, Lucrezia},
  title   = {Large {Bayesian} Vector Auto Regressions},
  journal = {Journal of Applied Econometrics},
  year    = {2010},
  volume  = {25},
  number  = {1},
  pages   = {71--92}
}

@techreport{BCBS2019,
  author      = {{Basel Committee on Banking Supervision}},
  title       = {Minimum Capital Requirements for Market Risk},
  institution = {Bank for International Settlements},
  year        = {2019},
  url         = {https://www.bis.org/bcbs/publ/d457.pdf}
}

@article{berman2012different,
  author  = {Berman, Nicolas and Martin, Philippe and Mayer, Thierry},
  title   = {How Do Different Exporters React to Exchange Rate Changes?},
  journal = {The Quarterly Journal of Economics},
  year    = {2012},
  volume  = {127},
  number  = {1},
  pages   = {437--492}
}

@article{borio2014stress,
  author  = {Borio, Claudio and Drehmann, Mathias and Tsatsaronis, Kostas},
  title   = {Stress-Testing Macro Stress Testing: Does It Live Up to Expectations?},
  journal = {Journal of Financial Stability},
  year    = {2014},
  volume  = {12},
  pages   = {3--15}
}

@article{caldara2022geopolitical,
  author  = {Caldara, Dario and Iacoviello, Matteo},
  title   = {Measuring Geopolitical Risk},
  journal = {American Economic Review},
  year    = {2022},
  volume  = {112},
  number  = {4},
  pages   = {1194--1225},
  doi     = {10.1257/aer.20191823}
}

@techreport{camara2015mercure,
  author      = {Camara, Boubacar and Castellani, Fran{\c c}ois-Daniel and Fraisse, Henri and Frey, Laure and H{\'e}am, Jean-Cyprien and Labonne, Claire and Martin, Vincent},
  title       = {Mercure: A Macroprudential Stress Testing Model Developed at the {ACPR}},
  institution = {Autorit{\'e} de Contr{\^o}le Prudentiel et de R{\'e}solution, Banque de France},
  type        = {D{\'e}bats {\'E}conomiques et Financiers},
  number      = {19},
  year        = {2015}
}

@unpublished{ChanPfarrhofer2025,
  author = {Chan, Joshua C. C. and Pfarrhofer, Michael},
  title  = {Large {Bayesian} {VARs} for Binary and Censored Variables},
  year   = {2025},
  note   = {arXiv:2506.01422}
}

@article{ChavleishviliManganelli2024,
  author  = {Chavleishvili, Sulkhan and Manganelli, Simone},
  title   = {Forecasting and Stress Testing with Quantile Vector Autoregression},
  journal = {Journal of Applied Econometrics},
  year    = {2024},
  volume  = {39},
  number  = {1},
  pages   = {66--85},
  doi     = {10.1002/jae.3009}
}

@article{dueker2005dynamic,
  author  = {Dueker, Michael J.},
  title   = {Dynamic Forecasts of Qualitative Variables: A {Qual} {VAR} Model of {U.S.} Recessions},
  journal = {Journal of Business \& Economic Statistics},
  year    = {2005},
  volume  = {23},
  number  = {1},
  pages   = {96--104},
  doi     = {10.1198/073500104000000613}
}

@techreport{EBA2025method,
  author      = {{European Banking Authority}},
  title       = {{Methodological Note for the 2025 EU-Wide Stress Test}},
  institution = {European Banking Authority},
  year        = {2025},
  url         = {https://www.eba.europa.eu/risk-and-data-analysis/risk-analysis/eu-wide-stress-testing},
  note        = {Final methodology, published 20 January 2025}
}

@misc{ECB_SSM_PR_2025,
  author       = {{European Central Bank}},
  title        = {{ECB} to Assess Banks' Stress Testing Capabilities to Capture Geopolitical Risk},
  year         = {2025},
  howpublished = {Banking Supervision press release, 12 December 2025},
  url          = {https://www.bankingsupervision.europa.eu/press/pr/date/2025/html/ssm.pr251212~69f656d4bf.en.html}
}

@techreport{FornariLemke2010,
  author      = {Fornari, Fabio and Lemke, Wolfgang},
  title       = {Predicting Recession Probabilities with Financial Variables over Multiple Horizons},
  institution = {European Central Bank},
  type        = {Working Paper Series},
  number      = {1255},
  year        = {2010}
}

@article{gonzalezrivera2024expecting,
  author  = {Gonz{\'a}lez-Rivera, Gloria and Rodr{\'i}guez-Caballero, C. Vladimir and Ruiz, Esther},
  title   = {Expecting the Unexpected: Stressed Scenarios for Economic Growth},
  journal = {Journal of Applied Econometrics},
  year    = {2024},
  volume  = {39},
  pages   = {926--942}
}

@article{Gordy2003,
  author  = {Gordy, Michael B.},
  title   = {A Risk-Factor Model Foundation for Ratings-Based Bank Capital Rules},
  journal = {Journal of Financial Intermediation},
  year    = {2003},
  volume  = {12},
  number  = {3},
  pages   = {199--232},
  doi     = {10.1016/S1042-9573(03)00040-8}
}

@incollection{HenryKok2013,
  author    = {Henry, J{\'e}r{\^o}me and Kok, Christoffer},
  title     = {A Macro Stress Testing Framework for Assessing Systemic Risks in the Banking Sector},
  booktitle = {Handbook of Systemic Risk},
  editor    = {Fouque, Jean-Pierre and Langsam, Joseph A.},
  publisher = {Cambridge University Press},
  year      = {2013},
  pages     = {619--640}
}

@unpublished{hurlin2026reverse,
  author = {Hurlin, Christophe and Lajaunie, Quentin and Pull, Yoann},
  title  = {Reverse Stress Testing Geopolitical Risk in Corporate Credit Portfolios: A Formal and Operational Framework},
  year   = {2026},
  note   = {arXiv:2601.03983}
}

@article{InoueKilian2022,
  author  = {Inoue, Atsushi and Kilian, Lutz},
  title   = {Joint {Bayesian} Inference about Impulse Responses in {VAR} Models},
  journal = {Journal of Econometrics},
  year    = {2022},
  volume  = {231},
  number  = {2},
  pages   = {457--476}
}

@article{kadiyala1997numerical,
  author  = {Kadiyala, K. Rao and Karlsson, Sune},
  title   = {Numerical Methods for Estimation and Inference in {Bayesian} {VAR}-Models},
  journal = {Journal of Applied Econometrics},
  year    = {1997},
  volume  = {12},
  number  = {2},
  pages   = {99--132}
}

@article{koop1996impulse,
  author  = {Koop, Gary and Pesaran, M. Hashem and Potter, Simon M.},
  title   = {Impulse Response Analysis in Nonlinear Multivariate Models},
  journal = {Journal of Econometrics},
  year    = {1996},
  volume  = {74},
  number  = {1},
  pages   = {119--147}
}

@book{koop2003bayesian,
  author    = {Koop, Gary},
  title     = {Bayesian Econometrics},
  publisher = {John Wiley \& Sons},
  address   = {Chichester},
  year      = {2003}
}

@article{laudati2023identifying,
  author  = {Laudati, Dario and Pesaran, M. Hashem},
  title   = {Identifying the Effects of Sanctions on the {Iranian} Economy Using Newspaper Coverage},
  journal = {Journal of Applied Econometrics},
  year    = {2023},
  volume  = {38},
  number  = {3},
  pages   = {271--294}
}

@article{merton1974pricing,
  author  = {Merton, Robert C.},
  title   = {On the Pricing of Corporate Debt: The Risk Structure of Interest Rates},
  journal = {The Journal of Finance},
  year    = {1974},
  volume  = {29},
  number  = {2},
  pages   = {449--470}
}

@article{pesaran2006macroeconomic,
  author  = {Pesaran, M. Hashem and Schuermann, Til and Treutler, Bj{\"o}rn-Jakob and Weiner, Scott M.},
  title   = {Macroeconomic Dynamics and Credit Risk: A Global Perspective},
  journal = {Journal of Money, Credit and Banking},
  year    = {2006},
  volume  = {38},
  number  = {5},
  pages   = {1211--1261}
}

@article{pesaranShin1998,
  author  = {Pesaran, M. Hashem and Shin, Yongcheol},
  title   = {Generalized Impulse Response Analysis in Linear Multivariate Models},
  journal = {Economics Letters},
  year    = {1998},
  volume  = {58},
  number  = {1},
  pages   = {17--29},
  doi     = {10.1016/S0165-1765(97)00214-0}
}

@book{Quagliariello2009,
  author    = {Quagliariello, Mario},
  title     = {Stress-Testing the Banking System: Methodologies and Applications},
  publisher = {Cambridge University Press},
  address   = {Cambridge},
  year      = {2009}
}

@article{vasicek2002distribution,
  author  = {Vasicek, Oldrich},
  title   = {The Distribution of Loan Portfolio Value},
  journal = {Risk},
  year    = {2002},
  volume  = {15},
  number  = {12},
  pages   = {160--162}
}

@techreport{Virolainen2004,
  author      = {Virolainen, Kimmo},
  title       = {Macro Stress Testing with a Macroeconomic Credit Risk Model for {Finland}},
  institution = {Bank of Finland},
  type        = {Discussion Paper},
  number      = {18/2004},
  year        = {2004}
}

@article{wang2021axiomatic,
  author  = {Wang, Ruodu and Zitikis, Ri{\v{c}}ardas},
  title   = {An Axiomatic Foundation for the Expected Shortfall},
  journal = {Management Science},
  year    = {2021},
  volume  = {67},
  number  = {3},
  pages   = {1413--1429},
  doi     = {10.1287/mnsc.2020.3617}
}

@article{Wilson1997,
  author  = {Wilson, Thomas C.},
  title   = {Portfolio Credit Risk},
  journal = {Risk},
  year    = {1997},
  volume  = {10},
  number  = {10},
  pages   = {56--61}
}

\end{document}